\documentclass[a4paper,twocolumn,11pt,accepted=2023-02-06]{quantumarticle}
\pdfoutput=1
\usepackage{graphicx}
\usepackage{epsfig}
\usepackage{amsfonts}
\usepackage{amsmath}
\usepackage{amssymb}
\usepackage{amsthm}

\usepackage{color}
\usepackage[colorlinks=true,linkcolor=blue,citecolor=magenta,urlcolor=blue]{hyperref}
\usepackage{comment}
\usepackage{braket}

\newtheorem{definition}{Definition}

\newtheorem{lemma}{Lemma}

\newtheorem{theorem}{Theorem}

\newcommand{\tr}{{\rm tr}}
\newcommand{\eprint}[1]{#1}

\begin{document}


\title{Graph-theoretic approach to Bell experiments with low detection efficiency}


\author{Zhen-Peng Xu}
\affiliation{School of Physics and Optoelectronics Engineering, Anhui University, 230601 Hefei, People’s Republic of China}
\affiliation{Naturwissenschaftlich-Technische Fakult\"at, Universit\"at Siegen, Walter-Flex-Stra{\ss}e 3, 57068 Siegen, Germany}

\author{Jonathan Steinberg}
\affiliation{Naturwissenschaftlich-Technische Fakult\"at, Universit\"at Siegen, Walter-Flex-Stra{\ss}e 3, 57068 Siegen, Germany}

\author{Jaskaran Singh}
\affiliation{Departamento de F\'{\i}sica Aplicada II, Universidad de Sevilla, E-41012 Sevilla, Spain}

\author{Antonio J. L\'opez-Tarrida}
\affiliation{Departamento de F\'{\i}sica Aplicada II, Universidad de Sevilla, E-41012 Sevilla, Spain}

\author{Jos\'e R. Portillo}
\affiliation{Departamento de Matem\'atica Aplicada I, Universidad de Sevilla, E-41012 Sevilla, Spain}
\affiliation{Instituto Universitario de Investigaci\'on de Matem\'aticas Antonio de Castro Brzezicki, Universidad de
Sevilla, E-41012 Sevilla, Spain}

\author{Ad\'an Cabello}
\email{adan@us.es}
\affiliation{Departamento de F\'{\i}sica Aplicada II, Universidad de Sevilla, E-41012 Sevilla, Spain}
\affiliation{Instituto Carlos~I de F\'{\i}sica Te\'orica y Computacional, Universidad de
Sevilla, E-41012 Sevilla, Spain}

\maketitle


\begin{abstract}
Bell inequality tests where the detection efficiency is below a certain threshold $\eta_{\rm{crit}}$ can be simulated with local hidden-variable models. 
Here, we introduce a method to identify Bell tests requiring low $\eta_{\rm{crit}}$ and relatively low dimension $d$ of the local quantum systems. The method has two steps. First, we show a family of bipartite Bell inequalities for which, for correlations produced by maximally entangled states, $\eta_{\rm{crit}}$ can be upper bounded by a function of some invariants of graphs, and use it to identify correlations that require small $\eta_{\rm{crit}}$. We present examples in which, for maximally entangled states, $\eta_{\rm{crit}} \le 0.516$ for $d=16$, $\eta_{\rm{crit}} \le 0.407$ for $d=28$, and $\eta_{\rm{crit}} \le 0.326$ for $d=32$. We also show evidence that the upper bound for $\eta_{\rm{crit}}$ can be lowered down to $0.415$ for $d=16$ and present a method to make the upper bound of $\eta_{\rm{crit}}$ arbitrarily small by increasing the dimension and the number of settings. All these upper bounds for $\eta_{\rm{crit}}$ are valid (as it is the case in the literature) assuming no noise. The second step is based on the observation that, using the initial state and measurement settings identified in the first step, we can construct Bell inequalities with smaller $\eta_{\rm{crit}}$ and better noise robustness. For that, we use a modified version of Gilbert's algorithm that takes advantage of the automorphisms of the graphs used in the first step. We illustrate its power by explicitly developing an example in which $\eta_{\rm{crit}}$ is $12.38\%$ lower and the required visibility is $14.62\%$ lower than the upper bounds obtained in the first step. The tools presented here may allow for developing high-dimensional loophole-free Bell tests and loophole-free Bell nonlocality over long distances.
\end{abstract}



\section{Introduction}


Bell nonlocality \cite{Bell64,BCPSW14}, that is, the violation of inequalities satisfied by any local hidden-variable (LHV) model, called Bell inequalities, is one of the most fascinating features of quantum mechanics and a crucial mean to accomplish tasks that are impossible with classical resources. 

The detection efficiency $\eta$ in an experimental test of a Bell inequality is the ratio between the number of systems detected by one party and the number of pairs emitted by the source. $\eta$ depends not only on the efficiency of the detectors, but also on all the losses occurring during the distribution of the state. Pearle \cite{Pearle70} and Wigner \cite{Wigner70} noticed that experimental correlations in Bell inequality tests where the detection efficiency is below a certain threshold $\eta_{\rm{crit}}$ can be simulated by LHV models. Therefore, if $\eta$ is not high enough, the quantum advantage in many Bell inequality-based protocols (e.g., for randomness expansion \cite{Colbeck06,PAMBMMOHLMM10,SZBS19,LLRZB19} and secure key distribution \cite{Ekert91,MY98,BHK05,ABGMPS07,PABGMS09}) vanishes. Avoiding this so-called ``detection loophole'' requires surpassing $\eta_{\rm{crit}}$. However, $\eta_{\rm{crit}}$ depends on the quantum correlations (i.e., the state prepared and the measurements performed) and the targeted Bell inequality. 

In this work, we identify quantum correlations and Bell inequalities requiring the smallest $\eta_{\rm{crit}}$ reported, to our knowledge, in the literature for symmetric Bell tests (i.e., those in which all detectors have the same detection efficiency). The correlations are produced by maximally entangled states of experimentally accessible dimensions (e.g., $d=16$) and two-outcome measurements. In addition, we describe a method that provides arbitrarily small $\eta_{\rm{crit}}$ but requires higher dimensions.

The importance of these results lies on the fact that they pave the way to extend the groundbreaking loophole-free Bell tests performed on quantum systems of dimension $d=2$ (i.e., qubits) and distances between $60$ and $1300$ m \cite{HBD15,GVW15,SMC15,RBG17} to quantum systems of higher dimensions and longer distances.


\subsection{Structure of the paper}


In Sec.~\ref{Prev}, we collect the smallest values of $\eta_{\rm{crit}}$ reported in the literature and explain how they have been obtained. We also recall a standard form of expressing Bell functionals that will be used in several places.

Finding states, measurements, and Bell inequalities leading to small values of $\eta_{\rm{crit}}$ is difficult. One of the reasons is that, as we will see, these values occur for Bell inequalities with many settings for which computing the local bound is, in general, intractable \cite{Pitowsky89}. Another reason is the difficulty of, given a Bell inequality, finding the state and measurements producing maximal quantum violations.

In Sec.~\ref{GBBI}, we introduce a method that overcomes both problems, as it connects the local bound of a family of bipartite Bell inequalities to the independence number of a graph, and quantum realizations that maximally violate the Bell inequalities to orthogonal representations of the graph.
This allows us to identify examples of quantum correlations and Bell inequalities with unprecedentedly low $\eta_{\rm{crit}}$.

In Sec.~\ref{sec:asde}, we present a method based on graph theory to obtain explicit states and measurements making $\eta_{\rm{crit}}$ arbitrarily close to zero.

There are, however, two problems that make the examples obtained up to that point not useful in practice. One is that, in some cases, they require too many settings. In Sec.~\ref{sec:search}, we present a method for searching for examples with low $\eta_{\rm{crit}}$ and smaller number of settings. 

The other problem is that, by construction, the quantum violations of the graph-based Bell inequalities are very sensitive to noise. In Sec.~\ref{sec:no}, we explain why and present a method to address this problem. The method applies to any of the (highly-symmetric) examples identified with the techniques in Sec.~\ref{GBBI} and produces a Bell inequality for the same correlations (i.e., the same state and measurements) with smaller $\eta_{\rm{crit}}$ and higher resistance to noise.

Finally, in Sec.~\ref{sec:alg}, we discuss how the different methods presented in this work can help us to design loophole-free Bell tests between high-dimensional quantum systems and achieve loophole-free Bell nonlocality over longer distances.


\section{Previous works}
\label{Prev}


\subsection{Existing values of $\eta_{\rm{crit}}$}
\label{sec:pw}


For symmetric Bell tests and perfect visibility (as defined later), $\eta_{\rm{crit}}=0.828$ using maximally entangled states \cite{GM87} for the simplest Bell inequality (the one with two dichotomic settings per party): the Clauser-Horne-Shimony-Holt (CHSH) inequality \cite{CHSH69}. Eberhard \cite{Eberhard93} noticed that $\eta_{\rm{crit}}$ can be further reduced using nonmaximally entangled states and, in particular, that it can be lowered down to $0.667$ for the Clauser-Horne (CH) inequality \cite{CH74}.

For Bell inequalities with four settings and maximally entangled states, $\eta_{\rm{crit}}$ is not better than for CHSH (with one exception that allows a slightly lower value; $\eta_{\rm{crit}} = 0.821$) \cite{BG08,ZG19}.

Although loophole-free Bell tests \cite{HBD15,GVW15,SMC15,RBG17} have proven that it is possible to produce correlations between local quantum systems of dimension $d=2$ that do not admit LHV models, the value of $\eta_{\rm{crit}}$ required in these experiments ($\ge 0.720$ due to the noise) prevents real-life applications and, in particular, applications outside laboratories with well controlled losses, or involving longer distances (e.g., $> 5$ km). It is in this sense that the requirement of $\eta_{\rm{crit}} > 0.667$ (without noise, $\eta_{\rm{crit}} > 0.720$ with noise) acts as a bottleneck that blocks many real-life applications.

Massar \cite{Massar02} discovered that high-dimensional systems can tolerate a detection efficiency that decreases with the dimension~$d$ of the local quantum system. However, he only found an improvement over the qubit case for $d > 1600$. V\'ertesi, Pironio, and Brunner \cite{VPB10} lowered $\eta_{\rm{crit}}$ down to $0.770$ for maximally entangled states and to $0.618$ for nonmaximally entangled states using $d=4$. 

M\'arton, Bene, and V\'ertesi \cite{MBV21} have studied the case of $N$ copies of the two-qubit maximally entangled state and local Pauli measurements which act in the corresponding qubit subsystems. They obtained the following upper bounds for $\eta_{\rm{crit}}$: $0.809$ for $N=2$ (which is equivalent to $d=4$), $0.740$ for $N=3$ (which is equivalent to $d=8$), and $0.693$ for $N=4$ (which is equivalent to $d=16$). More recently, Miklin {\em et al.} have reported a method that allows us to reduce $\eta_{\rm{crit}}$ down to $\approx 0.469$ for $d=512$ \cite{M22}.

In this article, we focus on the case of two parties and symmetric detection efficiency, since we see it as the closest one for realistic applications. However, it should be mentioned that another way of obtaining low $\eta_{\rm{crit}}$ is by considering Bell experiments involving $ \ge 3$ spatially separated parties \cite{LS01,CRV08,PVB12}. Still, the difficulty for preparing and distributing the required states makes the resulting values $\eta_{\rm{crit}}$ unpractical for actual applications. In addition, it should also be pointed out that $\eta_{\rm{crit}}$ can be reduced for some detectors at the expense of requiring $\eta_{\rm{crit}} \approx 1$ for other detectors \cite{CL07,BGSS07,Garbarino10,AQCFCT12}.


\subsection{Calculating $\eta_{\rm{crit}}$}
\label{ceta}


Here, we summarize the different ways for calculating $\eta_{\rm{crit}}$ that can be found in the literature for the case that the propagation losses and detection inefficiencies are the same for all the detectors. For a more detailed discussion, see Ref.~\cite{Larsson14}.

In an ideal Bell test, every run would end up with the two particles emitted by the source being detected, one at one party's site and the other at the other party's. However, in a real Bell test this may not be the case. As pointed out before, reasons for that are the existence of propagation losses and the imperfection of the detectors. Consequently, a detection at one site may not be accompanied with a detection at the other site and also, in some runs, both particles may be undetected. 

In some cases, in addition to the local losses and imperfect detectors, the particles are not emitted at well-known times, then the number of emitted pairs, and thus the number of runs of the Bell test, will be unknown. 

Below we summarize how $\eta_{\rm{crit}}$ can be calculated in each case.

{\em Case I.} The number of runs is known. This is the case in event-ready experiments \cite{HKO12,HBD15,RBG17} and in the proposals \cite{RL09,GPS10,branciard2011detection,CS12} and experiments \cite{MMGSFHCRJ16} with heralded detection.

{\em Case I. Strategy I.} We can associate the no-detection with a new outcome of the measurement and find a new Bell inequality with the same number of settings but one more outcome per setting than in the original Bell inequality. Then, the experiment will be free of the detection loophole as soon as the new inequality is violated. The problem is that finding this new Bell inequality may be difficult. 

For example, if we add a new outcome to all the measurements in the $(2,m,2)$ Bell scenario (i.e., the one with $2$ parties, $m$ settings per party, and $2$ outcomes), then we end up in the $(2,m,3)$ Bell scenario. Then, the number of deterministic LHV assignments changes from $2^{2m}$ to $3^{2m}$. Meanwhile, the dimension of the LHV assignments changes from $2m+m^2$ to $4m+4m^2$. To make clear the difficulties that finding such Bell inequalities involve, notice that, so far, we do not know the Bell inequalities for any $(2,m,3)$ Bell scenario. 

{\em Case I. Strategy II.} We can associate the no-detection with one of the outcomes of each measurement and use the original Bell inequality. Then, the experiment will be loophole-free as soon as the original Bell inequality is violated.

To obtain $\eta_{\rm crit}$ in this case, we can act as follows. In an ideal Bell test in which we could achieve the maximum quantum value $Q$ of a Bell expression $I$ whose bound for LHV models is $C$, and in which the experimental detection efficiency is $\eta$, the experimental value of $I$ would be
\begin{equation}
    I_{\rm exp} = \eta^2 Q + \eta(1-\eta)(Q_A + Q_B)+ (1-\eta)^2 X,
\end{equation}
where $Q_A$ is the value of $I$ resulting of what the parties output when Alice has detected the particle but not Bob (and similarly $Q_B$), and $X$ is the value that the parties output when both Alice and Bob have not detected the particles. Usually, the outputs are chosen such that $X=C$. The critical detection efficiency using this strategy is
\begin{equation}\label{eq:protocol1}
    \eta_{\rm crit} = \frac{2C - Q_A - Q_B}{C + Q - Q_A - Q_B}.
\end{equation}

This is the strategy used to obtain all the values previously reported except the ones in \cite{Eberhard93,VPB10}. This will also be the strategy used in this paper.

{\em Case II.} The number of runs is not known. This is the case in existing photonic Bell tests with high detection efficiency \cite{GMR13,CMA13,YXZ18,SLT18,BKG18,LZL18,LWZ18,ZSB20,SZB21,LZL21,LLR21,LZZ22}.

{\em Case II. Strategy III.} This strategy works for the CH Bell inequality \cite{CH74}, for inequalities that can be written in terms of the CH functional \cite{VPB10}, and for some Bell inequalities that can be rewritten similarly \cite{CCX13}. The CH functional has a peculiarity that makes it useful. It only contains joint and marginal probabilities of one of the outcomes. Then, there is no need to modify the bound for LHV models when $\eta$ is not $1$. Instead, a lower $\eta$ makes that the probabilities become lower so that the value of the Bell expression decreases. The joint probabilities decrease faster than the marginal probabilities. This strategy uses the expected values of $\eta$ and the noise for choosing the nonmaximally entangled state that maximizes the violation \cite{Eberhard93}. 

This is the strategy used in \cite{Eberhard93,VPB10} and adopted in the photonic loophole-free Bell tests \cite{GVW15,SMC15}.

{\em Case II. Other strategies.} These strategies compute $\eta_{\rm crit}$ under extra assumptions on the distribution of nondetections. See, e.g., \cite{GM87}. We do not enter into details here, as the need of extra assumptions is considered a weak point \cite{LS81,Larsson14}.


\subsection{Collins-Gisin parametrization}


Following the idea introduced in Ref.~\cite{CG04}, we will sometimes use a matrix to specify the functional $I$ associated to a Bell inequality $I \le C$ (where $C$ is the bound for LHV models) as
\begin{equation}
\begin{aligned}
&I = \\
&\left(
\setlength{\tabcolsep}{0.2em}
\begin{tabular}{c|ccc}
 & $c(\Pi_1^A\! =\! 1)$ & \dots & $c(\Pi_m^A\! =\! 1)$ \\
\hline 
$c(\Pi_1^B\! =\! 1)$ & $c(\Pi_1^A\! =\! \Pi_1^B\! =1)$ & $\dots$ & $c(\Pi_m^A\! =\! \Pi_1^B\! =\! 1)$ \\
$\vdots$ & $\vdots$ & $\ddots$ & $\vdots$ \\
$c(\Pi_m^B\! =\! 1)$ & $c(\Pi_1^A\! =\! \Pi_m^B\! =\! 1)$ & $\dots$ & $c(\Pi_m^A\! =\! \Pi_m^B\! =\! 1)$ \\
\end{tabular}
\right),
\label{eq:bell_functional}
\end{aligned}
\end{equation}
where the entries are coefficients for different terms that appear in $I$, which is assumed to contain only two-outcome ($0$ and $1$) measurements. As an 
example, $c(\Pi_1^A=1)$ indicates the coefficient that multiplies $P(\Pi_1^A=1)$ and 
$c(\Pi_1^A=\Pi_1^B=1)$ the coefficient that multiplies $P(\Pi_1^A=\Pi_1^B=1)$.

This technique allows us to write any Bell functional (even those with measurements with more than two outcomes) as a linear combination of joint and marginal probabilities, but without including one of the outcomes of each measurement. The coefficients of the probabilities involving this outcome can be computed from
the normalization and no-signalling conditions.


\section{Graph-theoretic approach to Bell nonlocality with low detection efficiency} 
\label{GBBI}


\subsection{General results}


Here, we first introduce a family of bipartite Bell inequalities, in which each inequality is associated to a graph $G$, such that the number of settings of each party coincides with the number of vertices of $G$ and the number of outcomes is two. The first interesting point about this family is that the LHV bound of each inequality coincides with the independence number of $G$. Therefore, we can take advantage of the vast literature on independence numbers of countless families of graphs to write Bell inequalities whose local bounds would be difficult to compute otherwise.


\begin{definition}[Independent set]
An independent set \cite{Diestel17} of a graph $G$ is a subset of vertices where any two vertices are nonadjacent.
\end{definition}

Hereafter, $u\sim v$ will indicate that $u$ and $v$ are adjacent.

\begin{definition}[Independence number]
The independence number \cite{Diestel17} of a graph $G$, denoted by $\alpha$, is the largest cardinality of any independent set of $G$.
\end{definition}


\begin{definition}[Xi number]
The xi number of a graph $G$ is
\begin{align}
\Xi = \min_{S\in \mathcal{S}_{\alpha+1}} \Xi(S), 
\end{align}
where $\mathcal{S}_{\alpha+1}$ is the set of all subsets of $(\alpha+1)$ vertices of $G$, where $\alpha$ is the independence number of $G$, and
\begin{align}
\Xi(S) := \max_{v\in S} \left\vert \{u|u\sim v, u\in S\} \right\vert.
\end{align}
\end{definition}


For example, $\Xi=2$ for the circulant graph $Ci_{10}(2,3)$, i.e., the $10$-vertex graph in which vertex $i$ is adjacent to vertices $i+2$ and $i+3$.


\begin{definition}[Circulant graph]
A graph with vertices $1,\ldots,|V|$ is circulant if the cyclic permutation $(1,\ldots,|V|)$ is a graph automorphism.
\end{definition}

\begin{definition}[Graph automorphism]
An automorphism of a graph $G$ is a permutation $\sigma$ of the vertex set of $G$, such that the pair of vertices $(i,j)$ is adjacent (i.e., forms an edge) if and only if the pair $(\sigma(i),\sigma(j))$ is adjacent. 
\end{definition}


By definition, if $S_1\subseteq S_2$, then
\begin{align}
\Xi(S_1) \le \Xi(S_2).
\end{align}
This implies that, if $S$ contains no fewer than $(\alpha+1)$ vertices, then
\begin{align}
\Xi(S) \ge \Xi.
\end{align}
In addition, $\Xi\ge 1$, since there is at least one edge among any set of $(\alpha+1)$ vertices, given that $\alpha$ is the independence number.


\begin{theorem}
\label{th:1}
Given a graph $G$ with vertex set $V$, edge set $E$, independence number $\alpha$, and xi number $\Xi$, the following is a Bell inequality:
\begin{align}
\label{bbi}
I &= \sum_{i \in V} P(\Pi^A_i=\Pi^B_i=1) - \nonumber\\ 
&\ \ \sum_{(i,j) \in E}\! \frac{1}{2\Xi}\! \left[P(\Pi^A_i\!=\!\Pi^B_j\!=\!1)\! +\! P(\Pi^A_j\!=\!\Pi^B_i\!=\!1) \right]\nonumber\\
&\overset{\rm LHV}{\le} \alpha,
\end{align}
where $P(\Pi^A_i=\Pi^B_j=1)$ is the probability that Alice obtains the outcome $1$ and Bob obtains the outcome $1$ when Alice measures the observable $\Pi^A_i$ (with possible outcomes $0$ and $1$) and Bob measures the observable $\Pi^B_j$ (with possible outcomes $0$ and $1$).
\end{theorem}


\begin{proof}
To obtain the upper bound of $I$ for LHV models, we only need to consider deterministic probability assignments. From the definition of $I$, it is easy to see that the bound cannot be less than $\alpha$.
Therefore, the bound can only be obtained when the events $[\Pi^A_i=\Pi^B_i=1]$ have been assigned the value $1$ for any $i\in S$, where $S$ contains no fewer than $\alpha$ vertices. 

Let us assume that $S$ contains no fewer than $\alpha+1$ vertices and let us call $v$ the vertex in $S$ such that 
\begin{equation}
\left\vert \{u|u\sim v, u\in S\} \right\vert = \Xi(S).
\end{equation}
By changing the assignment of the event $[\Pi^A_v=\Pi^B_v=1]$ to be $0$, the increment of $I$ is $-1+\frac{\Xi(S)}{\Xi}$. This is because $\forall i,j \in S$,
\begin{align}
P(\Pi^A_i=\Pi^B_j=1) = P(\Pi^A_j=\Pi^B_i=1) = 1
\end{align}
with our current assignment, especially for $i=v$ or $j=v$. 

Therefore, in the case that $S$ contains no fewer than $(\alpha+1)$ vertices, we can always set the assignment of one event $[\Pi^A_v=\Pi^B_v=1]$ to be $0$, such that the value of $I$ does not decrease. This implies that the upper bound can be obtained in the case that $S$ contains exactly $\alpha$ vertices, which implies that the upper bound can be no more than $\alpha$. Consequently, the upper bound for LHV models is exactly $\alpha$.
\end{proof}


There is a second reason why the Bell inequalities \eqref{bbi} are interesting for us. The reason is that they allow us to establish a one-to-one connection between a quantum value for $I$ and another graph invariant of $G$. Moreover, this connection also gives us the initial state and the local observables that provide the quantum value for $I$.


\begin{definition}[Orthonormal representation]
An orthonormal representation in $\mathbb{C}^d$ of a graph $G$ with vertex set $V$ is an assignment of a nonzero unit vector $|v_i \rangle \in \mathbb{C}^d$ to each $i \in V$ satisfying that $\langle v_i | v_j \rangle=0$ for all pairs $i,j$ of adjacent vertices. Such an assignment does not require that different vertices are assigned different vectors, nor that nonadjacent vertices correspond to nonorthogonal vectors.
\end{definition}

An additional unit vector $|\psi \rangle \in \mathbb{C}^d$, called {\em handle}, is sometimes specified together with the orthonormal representation.

Notice that in many works in graph theory the usual definition of orthonormal representation assigns orthogonal vectors to nonadjacent ---instead of adjacent--- vertices.

\begin{definition}[Orthogonal rank]
The orthogonal rank \cite{HPRS17} of a graph $G$, denoted $\xi$, is the smallest positive integer $d$ for which there is an orthonormal representation in $\mathbb{C}^d$ of $G$.
\end{definition}

Quantum pure states are represented by rays. Therefore, $\xi$ is also the minimum dimension a quantum system must have so adjacent vertices in $G$ can be assigned orthogonal quantum states (or orthogonal rank-one projectors). However, it can be the case that the same ray is assigned to different vertices.

\begin{definition}[Graph of orthogonality]
Given a set of vectors $S$, the graph of orthogonality of $S$ is the graph in which each vector is represented by a vertex and two vertices are adjacent if and only if their corresponding vectors are mutually orthogonal.
\end{definition}


\begin{theorem}
\label{th:2}
For any graph $G$, the maximum quantum value of $I$, defined in Eq.~\eqref{bbi}, is
\begin{equation}
\label{eq:3}
Q \ge \frac{|V|}{\xi},
\end{equation}
where $|V|$ is the number of vertices of $G$ and $\xi$ is the orthogonal rank of $G$. The value $I=\frac{|V|}{\xi}$
is achieved by preparing the maximally entangled state 
\begin{equation}
\label{eq:state}
|\psi\rangle = \frac{1}{\sqrt{\xi}}\sum_{j=0}^{\xi-1} |j\rangle |j\rangle
\end{equation}
and using as local settings on Alice's side the observables represented by the projectors $|v_i\rangle \langle v_i| \otimes \openone$, with $|v_i\rangle$ in an 
orthonormal representation of dimension $\xi$ of $G$, and as local settings on Bob's side the observables represented by the projectors $\openone \otimes |v_i^{\ast}\rangle \langle v_i^{\ast}|$, where $|v_i^{\ast}\rangle$ is the complex conjugate of $|v_i\rangle$.
\end{theorem}


\begin{proof}
By definition of $\xi$, the value $I=\frac{|V|}{\xi}$, can be achieved in quantum mechanics using the maximally entangled state \eqref{eq:state} and locally measuring the rank-one projectors corresponding to an orthonormal representation of dimension $\xi$ of $G$ in Alice's side and its complex conjugate in Bob's side.
\end{proof}


Theorems~\ref{th:1} and \ref{th:2} allow us to link an upper bound of the critical detection efficiency $\eta_{\rm crit}$ for the quantum violation of the Bell inequality \eqref{bbi} produced with maximally entangled states with invariants of the graph that originates the Bell inequality.


\begin{theorem}
\label{th:3}
For any Bell inequality of the form \eqref{bbi} associated to a graph $G$,
assuming that the number of runs of the Bell test is known (i.e., that we are in Case I in Sec.~\ref{ceta}) and that the parties adopt Strategy~II (described in Sec.~\ref{ceta}),
local models simulating the correlations produced by the state \eqref{eq:state} and the measurements described after Eq.~\eqref{eq:state} are impossible if the detection efficiency is
\begin{equation}
\label{eq:etac1}
\eta > \sqrt{\frac{\alpha}{|V|/\xi}} \ge \eta_{\rm crit},
\end{equation}
where $\alpha$, $|V|$, and $\xi$ are the independence number, the number of vertices, and the orthogonal rank of $G$, respectively.
\end{theorem}


\begin{proof}
Recall that the Collins-Gisin parametrization allows us to write any Bell expression as a linear combination of joint and marginal probabilities, without including one of the outcomes of each measurement. Then, a strategy in case of no detection is to associate the no-detection with the outcome $0$, which is assumed to be the one that does not appear explicitly in the Bell expression. 
Following this strategy, the probabilities in the Bell expression transforms as follows:
\begin{align}
    P(\Pi^A_i=1) \to \eta P(\Pi^A_i=1),\\
    P(\Pi^B_j=1) \to \eta P(\Pi^B_j=1),\\
    P(\Pi^A_i=\Pi^B_j=1) \to \eta^2 P(\Pi^A_i=\Pi^B_j=1).
\end{align}
If the Bell expression contains no marginal items, as it is the case in the Bell functional in \eqref{bbi}, then, in Eq.~\eqref{eq:protocol1}, $Q_A = Q_B = 0$. If no detection is associated to the outcome $0$, then $X = 0$.
Consequently, the quantum value in the ideal case becomes $\eta^2 Q$. Then, in this case,
\begin{equation}\label{eq:protocol2}
    \eta_{\rm crit} = \sqrt{\frac{C}{Q}},
\end{equation}
where $C$ is the upper bound of $I$ for LHV models.
Then, using Theorems \ref{th:1} and \ref{th:2}, we obtain Eq.~\eqref{eq:etac1}.
\end{proof}


\subsection{Examples of nonlocal correlations with low $\eta_{\rm crit}$}
\label{pnstates}


Here, we use Theorem~\ref{th:3} to identify quantum correlations and Bell inequalities with low $\eta_{\rm crit}$.


\subsubsection{Definitions}


\begin{definition}[Pauli observables]
The set $P_n$ of Pauli observables for a system of $n \ge 2$ qubits consists of the nontrivial quantum observables represented by $n$-term tensor products of the Pauli matrices $\sigma_x$, $\sigma_y$, $\sigma_z$, and $I$ (the $2 \times 2$ identity matrix).
\end{definition}

The cardinality of $P_n$ is $|P_n|=4^n-1$, since $P_n$ does not contain the $2^n\times 2^n$ identity matrix. 

\begin{definition}[Pauli states] \label{def:pauli_states}
The set ${\cal P}_n (\mathbb{C})$ of Pauli states for a system of $n \ge 2$ qubits consists of the common eigenstates of all the maximal subsets of $P_n$ containing only mutually compatible observables (i.e., represented by mutually commuting matrices) of maximal size. 
\end{definition}

The Pauli states are also called the ``quantum states arising from the Pauli group''~\cite{Planat13}. The eigenvectors of each subset of maximal size of $P_n$ containing only mutually compatible observables provide an unique orthogonal basis of vectors with $d=2^n$ vectors. There are $L=\prod_{j=1}^n (2^j + 1)$ such subsets, and ${\cal P}_n (\mathbb{C})$ is the union of the $L$ disjoint orthogonal bases. Accordingly, $|{\cal P}_n (\mathbb{C})|= L d$.

Hereafter, we denote by ${\cal P}_n (\mathbb{R})$ the subset of ${\cal P}_n (\mathbb{C})$ represented by vectors with all components in $\mathbb{R}$. $|{\cal P}_n (\mathbb{R})|=\prod_{j=1}^n (2^j + 2)$. 

\begin{definition}[Newman states]
The set ${\cal N}_d$ of Newman states for a quantum system of dimension $d$, where $d=4k, k \in \mathbb{N}$, consists of the states represented by $d$-dimensional rays with components $-1$ and $1$ and such that the number of $-1$ components is even.
\end{definition}

For example, 
${\cal N}_4 =\{(1,1,1,1)$, $(1,1,-1,-1)$,
$(1,-1,1,-1)$, 
$(1,-1,-1,1)\}$. We remark that $(a,b,c,d)$ and $-1\times (a,b,c,d)$ represent the same state. 
As it can be seen, $|{\cal N}_d|=2^{d-2}$.

\begin{definition}[Newman graphs]
The Newman graph $Y_d$ is the graph of orthogonality of ${\cal N}_d$, where $d=4k, k \in \mathbb{N}$.
\end{definition}

The name follows from a family of graphs studied by Newman (see Sec.~6.6 of \cite{Newman2004}).

\begin{definition}[Lov\'asz's number]
The Lov\'asz number of a graph $G$ is \cite{Lovasz79,GLS86}
\begin{equation}
\vartheta(G):= \max \sum_{i\in V} | \langle \psi | v_i \rangle |^2,
\label{Lovasz}
\end{equation}
where the maximum is taken over all orthonormal representations $\{| v_i \rangle \}_{i \in V}$ of $G$ and handles $| \psi \rangle$ (i.e., normalized vectors) in any dimension.
\end{definition}

By the definition of $\vartheta(G)$, there is a quantum realization that achieves $I=\vartheta(G)$ for the $I$ associated to $G$ using Eq.~\eqref{bbi}.

\begin{definition}[Fractional packing number]
The fractional packing number of a graph $G$ \cite{Lovasz79,GLS86,Shannon56} is
\begin{equation}
\alpha^{\ast}(G):=\max \sum_{i \in V} p_i,
\label{alphastar}
\end{equation}
where the maximum is taken over all $p_i \geq 0$ and for all cliques $C$ of $G$, under the restriction $\sum_{i \in C} p_i \leq 1$.
\end{definition}

It will be useful that $\vartheta(G) \le \alpha^{\ast}(G)$.

\begin{definition}[Hadamard or orthogonality graphs Omega]
For $n \in \mathbb{N}$, an orthogonality graph $\Omega_n = (V,E)$ is the graph with vertex set $V=\{-1,1\}^n$ and edge set $E=\{(u,v) \in V \times V: \langle u,v \rangle = 0 \}$. That is, each vertex is assigned a $\pm 1$-vector of length $n$, and two vertices are adjacent if and only if the corresponding vectors are orthogonal.
\end{definition}

Geometrically, the vectors assigned to the vertices of the Hadamard graph $\Omega_n$ correspond to the directions of the vertices of an $n$-dimensional hypercube centered at the origin. Newman states may be seen as a subset of such hypercube directions. Therefore, a Newman graph $Y_n$ is an induced subgraph of an orthogonality graph $\Omega_n$.

\begin{definition}[Induced subgraph]
Given a graph $G$ with vertex set $V$ and edge set $E$, and a subset of vertices $S \subset V$, the subgraph of $G$ induced by $S$ is the graph with vertex set $S$, and edge set consisting of all the edges $(u,v) \in E$ such that $u,v \in S$ \cite{Diestel17}.
\end{definition}

\begin{definition}[Hadamard matrix]
A Hadamard matrix of order $n$ is a real $n \times n$ square matrix $H_n$ in which all its entries are either $+1$ or $-1$, and whose rows are mutually orthogonal.
\end{definition}

The order $n$ of a Hadamard matrix must be $1$, $2$, or a multiple of $4$. Therefore, if $n$ is an even number, each pair of rows in a Hadamard matrix represents two mutually orthogonal $\pm 1$-vectors in dimension $n$. The same is true for its columns considered as $\pm 1$-vectors. Therefore, taking any pair of rows (alternatively, columns), the number of matching entries must be equal to the number of mismatching entries, exactly $n/2$.

\begin{definition}[Lexicographic product of graphs]
The lexicographic product of two graphs $G$ and $H$ with respective vertex sets $V(G)=\{u_i\}_i$ and $V(H)=\{v_k\}_k$ is a graph $G[H]$ such that its vertex set is the Cartesian product $V(G[H]) = V(G) \times V(H)$, and any two vertices $(u_i,v_k)$ and $(u_j,v_l)$ in $G[H]$ are adjacent if and only if either $u_i$ is adjacent with $u_j$ in $G$ or $u_i = u_j$ and $v_k$ is adjacent with $v_l$ in $H$. 
\end{definition}

The lexicographic product is associative but not commutative (a fact emphasized by the notation).


\subsubsection{Pauli-4320}


The graph of orthogonality of ${\cal P}_4 (\mathbb{R})$ has $\alpha = 72$ and $\vartheta=\alpha^{\ast}=\frac{|V|}{\xi}=270$. Therefore, by preparing the maximally entangled state \eqref{eq:state} of local dimension $\xi=2^4=16$ and allowing the parties to choose between the $4320$ two-outcome measurements represented by $|v_i\rangle \langle v_i|$, with $|v_i\rangle \in {\cal P}_4 (\mathbb{R})$, they produce a violation of the Bell inequality \eqref{bbi} which, using Theorem \ref{th:3}, allows us to conclude that
\begin{equation}
\label{eq:4320}
\eta_{\rm crit}^{{\cal P}_4 (\mathbb{R})} \le 0.516,
\end{equation}
which is an unprecedentedly low upper bound for this dimension (see Sec.~\ref{sec:pw}).

Notice that $4320$ local choices are not too many for a realistic Bell test. For example, a photonic loophole-free Bell test may have $3502784150$ trials \cite{GVW15}, which is enough for a Bell test in which each party has to chose between $4320$ settings, as it gives $187.7$ trials for each possible combination of settings $(x,y)$, which is more than three times the number of trials per $(x,y)$ in the first loophole-free Bell test \cite{HBD15}. Recall that all measurements have two outcomes, as in the test of the Clauser-Horne-Shimony-Holt Bell inequality \cite{CHSH69}. Therefore, only two detectors per party are necessary.


\subsubsection{Pauli-36720}


We conjecture that the graph of orthogonality of ${\cal P}_4 (\mathbb{C})$ has $\alpha = 396$. This conjecture is based on the fact that, after months of computations, $396$ is the largest value that we have found (and we have found it many times, which suggests that our search is sufficiently dense). The computation is based on a greedy-type algorithm taking into account the symmetry of the graph as well as known upper bounds for the independence number by means of spectral graph theory. Given a graph $G$, we proceed as follows. $(i)$~Compute the automorphism group of $G$ and the corresponding orbits. These orbits will yield a partition of the vertex set $\lbrace 1,\ldots, \vert V(G) \vert \rbrace$. $(ii)$~From each orbit $O_{k}$ we select a vertex $v^{k} \in O_{k}$. This vertex then has two options: It can either be part of a maximal independent vertex set or it is not, i.e., $v^{k} \in \mathcal{I} (G)$ or $v^{k} \notin \mathcal{I} (G)$. If $v^{k} \in \mathcal{I}$, then remove $v^{k}$ and all neighbours of $v^{k}$. This produces a tuple $(G^{1},1)$ containing a new graph ${G}^1$ and $1$, as we have removed a vertex from the independent set of the original graph. On the other hand, if $v^{k} \notin \mathcal{I}(G)$, we can remove the whole orbit $O_{k}$ of $G$ with $v^{k} \in O_{k}$ what yields another graph $G^{2}$. As we have not removed a member of $\mathcal{I}$, we store the tuple $(G^{2},0)$. In total, this produces a sequence of graphs with a strictly decreasing number of vertices. Once the size of all graphs is lower than a threshold $\kappa$ for which $\alpha$ can be computed directly, the decomposition stops yielding a set $\lbrace (G^{k}, H^{k}) \rbrace$, where $H^{k} \in \lbrace (0,1)^{n}\rbrace$ denotes the choice in each step. The independence number of $G$ is given by the maximum over $\Tilde{\alpha}_{k} = \alpha(G_{k}) + \sum_{j} (H^{k})_{j}$. However, the problem with this approach is (revealing the hardness of the problem of computing $\alpha$) that the number of graphs in the decomposition grows exponentially. If one has a sufficiently high lower bound for $\alpha(G)$ given a priori, one only needs to collect those graphs appearing in the decomposition process whose independence number is larger than this a priori bound.

In addition, $\vartheta=\alpha^{\ast}=\frac{|V|}{\xi}=2295$. Therefore, if the above conjecture is correct, then, by preparing the maximally entangled state \eqref{eq:state} of local dimension $\xi=2^4=16$ and allowing the parties to choose between the $36720$ two-outcome measurements represented by $|v_i\rangle \langle v_i|$, with $|v_i\rangle \in {\cal P}_4 (\mathbb{C})$, they can produce a violation of the Bell inequality \eqref{bbi} which, using Theorem \ref{th:3}, allows us to conclude that
\begin{equation}
\eta_{\rm crit}^{{\cal P}_4 (\mathbb{C})}\le0.415.
\end{equation}


\subsubsection{Newman-$2^{26}$}


As it is proven in Sec.~\ref{sec:proofs}, the graph of orthogonality of ${\cal N}_{28}$ has $\alpha = 397594$ and $\vartheta=\alpha^{\ast}=\frac{|V|}{\xi}=\frac{16777216}{7} \approx 2.3967 \times 10^6$. Therefore, by preparing the maximally entangled state \eqref{eq:state} of local dimension $\xi=28$ and allowing the parties to choose between the $2^{26}$ two-outcome measurements represented by $|v_i\rangle \langle v_i|$, with $|v_i\rangle \in {\cal N}_{28}$, they can produce a violation of the Bell inequality \eqref{bbi} which, using Theorem \ref{th:3}, allows us to conclude that
\begin{equation}
\eta_{\rm crit}^{{\cal N}_{28}}\le 0.407.
\end{equation}

Arguably, $2^{26}$ two-outcome measurements are too many for a real Bell test. The aim of this and the next example is to show that, by digging in the literature, one can find sets of vectors (or graphs), leading to gedanken Bell tests with very low $\eta_{\rm crit}$.


\subsubsection{Newman-$2^{30}$}


As it is proven in Sec.~\ref{sec:proofs}, the graph of orthogonality of ${\cal N}_{32}$ has $\alpha = 3572224$ and $\vartheta=\alpha^{\ast}=\frac{|V|}{\xi}=2^{25}$. Therefore, by preparing the maximally entangled state \eqref{eq:state} of local dimension $\xi=32$ and allowing the parties to choose between the $2^{30}$ two-outcome measurements represented by $|v_i\rangle \langle v_i|$, with $|v_i\rangle \in {\cal N}_{32}$, they produce a violation of the Bell inequality \eqref{bbi} which, using Theorem \ref{th:3}, allows us to conclude that
\begin{equation}
\eta_{\rm crit}^{{\cal N}_{32}} \le 0.326.
\end{equation}


\subsubsection{Independence number and quantum value for the Newman graphs}
\label{sec:proofs}


The independence number of the Newman graph $Y_n$ can be obtained by exploiting a connection \cite{Newman2004} between $Y_n$ and the orthogonality graphs $\Omega_n$, also known as Hadamard graphs or Deutch-Jozsa graphs. The graphs $\Omega_n$ were introduced by Ito \cite{Ito85a,Ito85b} as a tool to provide an algebraic graph theoretic background for Hadamard matrices. Hadamard graphs appear in relation to some quantum communication protocols and some proofs of the Kochen-Specker theorem \cite{BCT99,GTW2013,AHKS2005,CMNSW2007,SS2012}.

By definition, the graph $Y_n$ is a subgraph of $\Omega_n$ induced by a specific subset of its vertices.

For our purposes, the only interesting graphs $\Omega_n$ are those for which $n=4k, k \in \mathbb{N}$. Otherwise, $ \Omega_n$ is empty for $n$ odd, or bipartite for $n = 2 \mod 4$ \cite{Newman2004}. Restricting ourselves to such interesting graphs $\Omega_{n=4k}$, the first observation is that $\Omega_n$ is the disjoint union of two isomorphic graphs, 
\begin{equation}
\Omega_n = \Omega_n^{\rm e} \sqcup \Omega_n^{\rm o},
\end{equation}
where $\Omega_n^{\rm e}$ is the graph defined by the vertices corresponding to vectors with an even number of components $1$, and $\Omega_n^{\rm o}$ is the graph defined by the vertices corresponding to vectors with an odd number of components $1$.
Therefore, the independence numbers are related as follows:
\begin{equation}
\label{al1}
\alpha(\Omega_n) = \alpha(\Omega_n^{\rm e})+\alpha(\Omega_n^{\rm o}) = 2 \alpha(\Omega_n^{\rm e})
\end{equation}
and the orthogonal ranks are related as follows:
\begin{equation}
\xi(\Omega_n) = \xi(\Omega_n^{\rm e}) = \xi(\Omega_n^{\rm o}).
\end{equation}

The second step is noticing that $\Omega_n^{\rm e}$ is the lexicographic product of $Y_n$ with the complement of the complete graph on two vertices,
\begin{equation}
\Omega_n^{\rm e} = Y_n[\bar{K}_2].
\end{equation}
Therefore (see Theorem~\ref{th:ar}),
\begin{equation}
\label{al2}
\alpha(\Omega_n^{\rm e}) = \alpha(Y_n) \alpha(\bar{K}_2) = 2 \alpha(Y_n)
\end{equation}
and 
\begin{equation}
\xi(\Omega_n^{\rm e}) = \xi(Y_n).
\end{equation}

The orthogonal rank of $\Omega_n$ is $n$ \cite{MR2016,WE2018}. Therefore,
\begin{equation}
\xi(Y_n)=n,
\end{equation}
and the assignment of $n$-dimensional rays with components $-1$ and $1$ to the vertices $Y_n$ such that adjacent vertices are assigned orthogonal rays yields an orthogonal representation of $Y_n$ of minimum dimension. 

On the other hand, $\alpha(\Omega_n)$ is known for $n = 4p^k$, for $k \ge 1$ where $p$ is an odd prime \cite{Frankl86}, and also for $n=2^k$ for $k \ge 2$ \cite{IT2019}. In both cases,
\begin{equation}\label{alphaom}
\alpha(\Omega_n) = 4 \sum_{i=0}^{n/4-1}\binom{n-1}{i}.
\end{equation}
It still remains a conjecture whether Eq.~\eqref{alphaom} is valid when $n$ is another multiple of $4$. To our knowledge, the first open case is $n=40$ \cite{IT2019}.
Taking Eqs.~\eqref{al1}, \eqref{al2}, and \eqref{alphaom} into account,
\begin{equation}
\alpha(Y_{28}) = 397594
\end{equation}
and
\begin{equation}
\alpha(Y_{32}) = 3572224.
\end{equation}

In addition,
$Y_n$ has $|V|=2^{n-2}$ vertices and $|E|= 2^{n-4} \binom{n}{n/2}$ edges.

Let us show that, for the two considered Newman graphs, the orthogonal rank (i.e., the minimal dimension of the physical realization) equals their clique number (size of the largest clique).

In order to prove that Newman's graphs $Y_{28}$ and $Y_{32}$ contain cliques of size $28$ and $32$, respectively, note that such cliques correspond to sets of pairwise orthogonal $\pm 1$-rays of cardinality $28$ in dimension $28$, or cardinality $32$ in dimension $32$, in which the number of $-1$ components is an even (alternatively, odd) number. This fact allows us to rephrase this problem in a slightly different and more convenient way, using Hadamard matrices: our goal is to construct adequate Hadamard matrices $H_n$ of orders $n=28$ and $n=32$. Each row in $H_n$ is a $\pm 1$-vector in dimension $n$ and, by definition, the $n$ rows in $H_n$ constitute a set of $n$ pairwise orthogonal $\pm 1$-vectors in dimension $n$. In fact, these vectors are rays since no two rows can have the same entries with opposite signs, due to orthogonality. If necessary, we can transform $H_n$ into another equivalent $n \times n$ Hadamard matrix by negating rows or columns, or by interchanging rows or columns, so that the number of $-1$ components of the row vectors is an even (alternatively, odd) number. Notice that, in the end, the resulting set of vectors corresponds to a maximum clique of size $n$ in the Newman graph $Y_n$.

According to Hadamard's conjecture, a Hadamard matrix $H_n$ of order $n=4k$ exists for every positive integer $k$. At the present time, after applying the construction methods due to Sylvester, Paley, Williamson and others, the smallest order for which no Hadamard matrix is known is $n=668$. And there are many orders $n>668$ for which $H_n$ is known. This means that all Newman graphs $\Omega_{n=4k}$ with $n<668$ satisfy that $\omega(Y_n)=n$.

There is a well known recursive procedure to construct Hadamard matrices $H_n$ of order $n= 2^k, k \in \mathbb{N}$, the so called Sylvester's construction \cite{Sylvester1867}. Applying this procedure, $H_{32}$ can be obtained. This matrix fulfills the condition that the number of $-1$ entries in each row is an even number, hence providing a clique of size $32$ for the Newman graph $Y_{32}$.

Specifically, from $H_{32}$, we arrive at the following clique of size $32$: the set of rays of the form $u_i \otimes u_j \otimes v_k$, where $u_i, u_j \in \{(1,1,1,1), (1,1,-1,-1), (1,-1,1,-1)$, $(1,-1,-1,1)\}$, $v_k \in \{(1,1),(1,-1)\}$, and $\otimes$ denotes tensor product.

A Hadamard matrix $H_{28}$ is more convoluted. It can be obtained through the so-called Paley's construction (Lemma~2 in Ref.~\cite{Paley33}). There are $487$ inequivalent matrices $H_{28}$. Examples of them can be found in the literature. To exhibit a specific instance of a clique of size $28$ induced in $Y_{28}$ we look for a matrix $H_{28}$ such that the number of $-1$ entries in each row is again an even number. Such a matrix (using $0,1$ entries instead of $\pm 1$) can be found, v.g., in Fig.~1 in Ref.~\cite{Kimura94}: The set of row vectors obtained by replacing therein each $0$ entry with $-1$ constitutes the desired clique.

Finally, we will prove that for Newman graphs $Y_n$ with $n=28$ and $n=32$, the quantum value of $I$ given by Eq.~\eqref{bbi} can be $\alpha^{\ast}(Y_n)=\vartheta(Y_n)=\frac{|V(Y_n)|}{\xi(Y_n)}$. First, note that both $\Omega_n^{\rm e}$ and $\bar{K}_2$ are vertex-transitive. Given that $\Omega_n^{\rm e} = Y_n[\bar{K}_2]$, we know that $Y_n$ is also vertex-transitive, since the lexicographic product of two graphs is vertex-transitive if and only if both graph factors are vertex-transitive \cite{B80}.

On one hand, it is known \cite{MR2016} that
\begin{equation}
\vartheta(\Omega_n)=\frac{2^n}{n}.
\end{equation}
Since $\vartheta$ is multiplicative in the lexicographic product (see Theorem~\ref{th:ar}), we have $\vartheta(\Omega^{\rm e}) = \vartheta(Y_n) \times \vartheta(\bar{K}_2) = 2\, \vartheta(Y_n)$. Notice that $\vartheta(\Omega^{\rm e}) = \frac{\vartheta(\Omega_n)}{2}$, because $\Omega_n = \Omega^{\rm e} \sqcup \Omega^{\rm o}$. As a consequence,
\begin{equation}
\vartheta(Y_n)=\frac{\vartheta(\Omega_n)}{4}=\frac{2^{n-2}}{n}.
\end{equation}

On the other hand, the fractional packing number in a vertex-transitive graph $G$ satisfies $\alpha^{\ast}(G)=\frac{|V(G)|}{\omega(G)}$, where $|V(G)|$ is the number of vertices and $\omega(G)$ is the clique number of $G$. Given that $|Y_n|=2^{n-2}$ and knowing that the clique number for $Y_{28}$ and $Y_{32}$ is $\omega(Y_{28})=28$ and $\omega(Y_{32})=32$ (as proved before), and by vertex-transitivity, we obtain that the quantum values of $I$ can be
\begin{equation}
\alpha^{\ast}(Y_{28})=\vartheta(Y_{28})
\end{equation}
and
\begin{equation}
\alpha^{\ast}(Y_{32})=\vartheta(Y_{32}),
\end{equation}
respectively.


\subsubsection{Newman-$2^{26}$ and Newman-$2^{30}$ are state-independent contextuality sets}
\label{sec:sicn}


\begin{definition}[SI-C set]
A State-independent contextuality (SI-C) set \cite{CKB15} in dimension $d \ge 3$ is a set of projectors that produces noncontextual correlations (i.e. that violate some noncontextualty inequality) for any initial quantum state of dimension $d$.
\end{definition}

SI-C sets play a fundamental role in our method for identifying correlations with low $\eta_{\rm crit}$, as any SI-C set produces a quantum violation of a graph-based Bell inequality of the form \eqref{bbi}. However, there are sets that are not SI-C sets and produce a quantum violation of a Bell inequality of the form \eqref{bbi} \cite{CKB15}. 

\begin{theorem}
\label{th:sicn}
Newman-$2^{26}$ is a SI-C set in dimension $d=28$ and Newman-$2^{30}$ is a SI-C set in dimension $d=32$.
\end{theorem}

In order to prove Theorem~\ref{th:sicn}, we will first state and prove the following lemma, in which in a mild abuse of notation, we will use $\Omega_{n}^{\rm e}$ and $\Omega_{k}^{\rm o}$ to refer not only to the Hadamard graphs but also to the sets of vectors constituting their orthogonal representations:

\begin{lemma}
\label{lemma:sicn}
For $n\ge 3$,
\begin{equation}
\sum_{\langle v| \in \Omega_n^{\rm e}} |v\rangle\langle v| = \sum_{\langle v| \in \Omega_n^{\rm o}} |v\rangle\langle v| = 2^{n-1} \mathbb{I}_n,
\end{equation}
\begin{equation}
\sum_{\langle v| \in \Omega_n^{\rm e}} \langle v| = \sum_{\langle v| \in \Omega_n^{\rm o}} \langle v| = \langle 0|_n,
\end{equation}
where $\langle 0|_n = (0,\ldots,0)$ and all the vectors $\langle v|$'s are unnormalized.
\end{lemma}

\begin{proof}
We prove the lemma by induction. It's straightforward to verify that those claims hold for $n=3$. Let us assume now that they also hold for $n=k$. Notice that, by adding an adequate extra component $\pm 1$ to the vectors of $\Omega_{k}^{\rm e}$ and $\Omega_{k}^{\rm o}$, we obtain orthogonal representations $\Omega_{k+1}^{\rm e}$ and $\Omega_{k+1}^{\rm o}$, so that
\begin{align}
\Omega_{k+1}^{\rm e} &= \{ (\langle u|, 1) \vert \langle u| \in \Omega_{k}^{\rm e}\} \cup \{ (\langle u|, -1) \vert \langle u| \in \Omega_{k}^{\rm o}\},\\
\Omega_{k+1}^{\rm o} &= \{ (\langle u|, 1) \vert \langle u| \in \Omega_{k}^{\rm o}\} \cup \{ (\langle u|, -1)\langle u| \in \Omega_{k}^{\rm e}\}.
\end{align}

This implies that
\begin{align}
\sum_{\langle v| \in \Omega_{k+1}^{\rm e}} \langle v| &= \sum_{\langle u| \in \Omega_{k}^{\rm e}} (\langle u|,1) + \sum_{\langle u| \in \Omega_{k}^{\rm o}} (\langle u|,-1) \nonumber \\ 
&= (\langle 0|_k,2^{k-1}) + (\langle 0|_k,-2^{k-1}) \nonumber \\ 
&= \langle 0|_{k+1},
\end{align}
where the last equality holds because $\Omega_{k}^{\rm e}$ and $\Omega_{k}^{\rm o}$ have same number of elements, i.e., $2^{k-1}$.

On the other hand:
\begin{align}
&\sum_{\langle v| \in \Omega_{k+1}^{\rm e}} |v\rangle\langle v| \nonumber \\ =& \sum_{\langle u| \in \Omega_{k}^{\rm e}} \begin{bmatrix} |u\rangle\langle u| & |u\rangle\\ \langle u| & 1 \end{bmatrix} + \sum_{\langle u| \in \Omega_{k}^{\rm o}} \begin{bmatrix} |u\rangle\langle u| & -|u\rangle\\ -\langle u| & 1 \end{bmatrix} \nonumber \\ 
= &\begin{bmatrix} \sum\limits_{\langle u| \in \Omega_{k}^{\rm e}} |u\rangle\langle u| & |0\rangle_k \nonumber \\ \langle 0|_k & 2^{k-1} \end{bmatrix} + \begin{bmatrix} \sum\limits_{\langle u| \in \Omega_{k}^{\rm o}} |u\rangle\langle u| & |0\rangle_k\\ \langle 0|_k & 2^{k-1} \end{bmatrix} \nonumber \\
= & \begin{bmatrix} 2^{k-1} \mathbb{I}_{k} & |0\rangle_k\\ \langle 0|_k & 2^{k-1} \end{bmatrix} + \begin{bmatrix} 2^{k-1} \mathbb{I}_{k} & |0\rangle_k\\ \langle 0|_k & 2^{k-1} \end{bmatrix} \nonumber \\
= & 2^{k} \mathbb{I}_{k+1}.
\end{align}
Similarly, we can prove $\sum_{\langle v| \in \Omega_{k+1}^{\rm o}} |v\rangle\langle v| = 2^{k} \mathbb{I}_{k+1}$.

Thus, our claims hold for any $n\ge 3$.
\end{proof}

Now, Theorem \ref{th:sicn} is straightforward:

\begin{proof}
Let ${\cal N}_{n}$ be the set of rays constituting an orthonormal representation for the Newman graph $Y_n$. From Lemma~\ref{lemma:sicn}, and by definition,
\begin{equation}
\sum_{\langle v| \in {\cal N}_{n}} |v\rangle\langle v| = \frac{1}{2n} \sum_{\langle u| \in {\Omega}^{\rm e}_{n}} |u\rangle\langle u| = \frac{2^{n-2}}{n} \mathbb{I}_n,
\end{equation}
where $\langle v|$'s are normalized vectors, $\langle u|$'s are unnormalized. In the case $2^{n-2}/n > \alpha(Y_n)$, the set ${\cal N}_{n}$ is a SI-C set. In particular, this is true for ${\cal N}_{28}$ and ${\cal N}_{32}$, as claimed. It is also true for ${\cal N}_{n}$ with $n=12$, $16$, $20$, $36$, $44$, $52$, $64$, $68$, $100$, $108$, $128$, $196$, $256$, $324$, $484$, $500$, $512$, since these are the values for which Eq.~\eqref{alphaom} can be proven, satisfy $\alpha(Y_n) < |V(Y_n)|/\omega(Y_n)$, and are smaller than $668$, which is the smallest value for which no Hadamard matrix is known.
\end{proof}

For these sets of Newman states, $\eta_{\rm crit}$ tends to zero as $n$ grows. In particular, for $n=512$, $\eta_{\rm crit} < 1.6 \times 10^{-14}$.


\subsubsection{Experimental realization}


One way of preparing and measuring Pauli and Newman states of dimension $d$ is by using single photons (or neutrons or atoms or any type of radiation) in a Reck {\em et al.} $d$-input $d$-outcome multiport \cite{reck1994} or in its simplification by Clements {\em et al.} \cite{clements2016}. For state preparation, we may take further advantage from the fact that any unitary can be achieved no matter in which input the photon is injected. Therefore, we can use a specific input and remove all the elements in the paths not used. More interestingly is the possibility of simultaneously injecting several indistinguishable particles (either bosons or fermions) in different ports of a multiport interfeometer. This would allow us to achieve high $d$ using more compact setups. In this case, the problem of preparing and measuring Pauli and Newman states is still open, but can be addressed by taking advantage of the criteria for the suppression of certain output events in particular interferometers \cite{tichy2010,crespi2015}. For example, Newman states seem to be achievable using Sylvester interferometers \cite{crespi2015}.


\section{Arbitrarily small detection efficiency}
\label{sec:asde}


Here we show that, beyond specific examples, there are constructive methods such that, if there are no restrictions on the number of local settings or the local dimension of the quantum system, we can identify quantum correlations and a corresponding Bell inequality with respect to which the critical detection efficiency (above which no LHV model can be constructed) is as close to zero as desired.


\subsection{Definitions}


\begin{definition}[OR product of graphs]
The OR product (aka disjunctive product or conormal product) of two graphs $G$ and $H$ with respective vertex sets $V(G)=\{u_i\}_i$ and $V(H)=\{v_k\}_k$ is a graph $G \star H$ such that its vertex set is the Cartesian product $V(G \star H) = V(G) \times V(H)$, and any two vertices $(u_i,v_k)$ and $(u_j,v_l)$ in $G \star H$ are adjacent if and only if $u_i$ is adjacent with $u_j$ in $G$ or $v_k$ is adjacent with $v_l$ in $H$. The OR product is both associative and commutative.
\end{definition}

\begin{definition}[Spanning subgraph]
Given a graph $G$ with vertex set $V(G)$ and edge set $E(G)$, a spanning subgraph $H$ of $G$ (also known as a factor of $G$) is a subgraph of $G$ such that $V(G)=V(H)$ \cite{Diestel17}.
\end{definition}


\subsection{General results} 


\begin{theorem}
\label{th:ar}
If $G \circ H$ is the graph obtained either by the OR product or the lexicographic product of the graphs $G$ and $H$, then $\alpha(G \circ H)=\alpha(G)\alpha(H)$ and $\vartheta(G \circ H) = \vartheta(G)\vartheta(H)$.
\end{theorem}

That $\alpha$ and $\vartheta$ are both multiplicative in the OR product is proven in, e.g., \cite{Knuth94} (Sec.~21) and \cite{NR96} (Lemma~2.9). The same fact with respect to the lexicographic product is proven in, e.g., \cite{GS75} and \cite{Roberson2016}.

\begin{theorem}
\label{th:5}
For any Bell inequality of the form \eqref{bbi} associated to a graph $G^n$, denoting the OR or lexicographic product of $n$ copies of the graph $G$, local models are impossible if the detection efficiency is 
\begin{equation}
\label{eq:etac}
\eta > \sqrt{ \frac{\alpha^n}{(|V|/\xi)^n}} \ge \eta_{\rm crit}^{G^n},
\end{equation}
where $\alpha$, $|V|$, and $\xi$ are the independence number, the number of vertices, and the orthogonal rank of $G$, respectively.
\end{theorem}

The proof follows from Theorem~\ref{th:3} and Theorem~\ref{th:ar} for the case $H=G$. These two theorems can be combined whenever $\vartheta(G)=\alpha^{\ast}(G)=\frac{|V(G)|}{\xi(G)}$ and $\vartheta(H)=\alpha^{\ast}(H)=\frac{|V(H)|}{\xi(H)}$, but not in general.

By definition, $G[H]$ is a spanning subgraph of $G \star H$. More explicitly, $V(G[H]) = V(G \star H)$ and $E(G[H]) \subset E(G \star H)$. Therefore, by taking $G \star H$ and deleting some specific edges, we obtain $G[H]$. This implies that using the lexicographic product is more convenient, as a smaller number of edges in the final graph means that the violation of the Bell inequality \eqref{bbi} is, as we will see, more resistant to noise.


\subsection{Example: Pauli-$240^n$}


The graph of orthogonality of ${\cal P}_3 (\mathbb{R})$ has $\alpha = 16$ and $\vartheta=\alpha^{\ast}=30$. Therefore, by preparing the maximally entangled state \eqref{eq:state} of local dimension $\xi=2^3=8$ and allowing each of the parties to choose between the $240$ two-outcome measurements represented by $|v_i\rangle \langle v_i|$, with $|v_i\rangle \in {\cal P}_3 (\mathbb{R})$, they can produce a violation of the Bell inequality \eqref{bbi} which, using Theorem \ref{th:3}, allows us to conclude that $\eta_{\rm crit} \le 0.730$.

Therefore, with a system of local dimension $8^2$, and locally measuring the observables associated to the vertices of the lexicographic product of ${\cal P}_3 (\mathbb{R})$ with itself,
\begin{equation}
\eta_{\rm crit}^{{\cal P}_3^2 (\mathbb{R})} \le 0.533.
\end{equation}
And with a system of local dimension $8^3$, 
\begin{equation}
\eta_{\rm crit}^{{\cal P}_3^3 (\mathbb{R})} \le 0.389.
\end{equation}

The interest of this method is that it tends faster to $\eta_{\rm crit}=0$ using smaller $d$ than in any previous method. The downside is that, at least applied to the examples provided here, it requires too many settings.


\section{How to search for examples with low $\eta_{\rm crit}$ and a smaller number of settings}
\label{sec:search}


Most of the examples we have presented so far require too many settings to be tested in actual experiments. This leads to the question of whether we can achieve low $\eta_{\rm crit}$ using a moderate number (e.g., $< 100$) of settings. The aim of this section is showing that, arguably, the answer is affirmative. However, finding them will require some extra work. 


\subsection{First strategy: Vertex-transitive graphs}


So far, our strategy for finding examples with low $\eta_{\rm crit}$ was inspired by (the graphs of orthogonality of) sets of states common in quantum mechanics and quantum information (Pauli and Newman states). Interestingly, all our examples are sets of states whose graph of orthogonality is a vertex-transitive graph. In addition, vertex transitivity will be an important property for the second part of the paper, where optimizations of the Bell inequalities will be carried out. In fact, for a large number of settings, optimization will only be feasible for vertex-transitive graphs. 


\begin{definition}[Vertex-transitive graph]
A graph is vertex-transitive if, for every pair of vertices, there exists an automorphism of the graph mapping one to the other.
\end{definition}


Consequently, in searching for a systematic method to identify additional examples with low $\eta_{\rm crit}$ (for maximally entangled states and before any optimization), it makes sense to focus on vertex-transitive graphs.

Vertex-transitive graphs have been investigated for decades. As a fruit of these efforts (see, e.g., \cite{B80,MR90,RP89}), there are databases with all vertex-transitive graphs with up to $47$~vertices \cite{GD20}, all vertex-transitive graphs of degree~$3$ (i.e., each vertex is adjacent to three others) up to $1280$~vertices \cite{PSV13}, all circulant graphs up to $60$ vertices, and all circulant graphs with degrees at most $20$ up to $65$~vertices, at most $16$ up to $70$~vertices, and at most $12$ up to $100$~vertices~\cite{CombData}. Therefore, we can use these databases to compute $\eta_{\rm crit}$ for all these graphs and their complements and then select those that are interesting.

For any graph $G$, $\omega(G) \leq \vartheta(\overline{G}) \leq \xi(G) \leq \chi(G)$, where $\omega(G)$, $\vartheta(\overline{G})$, $\xi(G)$, and $\chi(G)$ are, respectively, the clique number, the Lov\'asz number of the complement of $G$, the orthogonal rank, and the chromatic number \cite{Knuth94, Lovasz2019}. The clique number $\omega$ is a trivial lower bound for $\xi$. The problem is that $\xi$ cannot be computed efficiently. However, in all the examples with low $\eta_{\rm crit}$ that we have identified, $\xi=\omega$. Therefore, we can use the databases and compute, for each $|V|$, the minimum of $\sqrt{\frac{\alpha}{|V|/\omega}}$. This gives a lower bound for $\eta_{\rm crit}$ that can be expected (for maximally entangled states and before any optimization) for the corresponding set of graphs. The results of these computations for all vertex-transitive graphs up to $47$~vertices are presented in Table~\ref{tab:vt_min}.

Table~\ref{tab:vt_min} shows that the aforementioned lower bound for $\eta_{\rm crit}$ decreases as the number of vertices increases. Moreover, it suggests that (for maximally entangled states and before any optimization) there are examples with $\eta_{\rm crit} < 0.5$ and $|V|<100$ vertices.

We can use existing computational tools~\cite{xu2021state} to estimate the exact $\xi$.
To find a orthogonal representation in $\mathbb{R}^d$ (or in $\mathbb{C}^d$) with minimal $\xi$ of the promising graphs, we can write each vector in the orthogonal representation as a unit vector using $d$ (or $2d$) real variables, and rotating the orthogonal representation into some canonical position to reduce the number of variables. Then, we take into account that the automorphisms of the graph (which can be easily computed) lead to geometric symmetries in the orthogonal representation. Then, using numerical optimization software, we run the minimization problem where the objective is to minimize the sum of squares of inner products for $\mathbb{R}^d$ (or the sum of squares of absolute values of inner products for $\mathbb{C}^d$), where the sum is taken over those pairs of vectors that are supposed to be orthogonal in the orthogonal representation. Notice that the automorphisms of the graph dramatically reduce the number of variables in the optimization problem because now we need only one vector per each orbit of a symmetry group. This works with many dozens of variables quite well. Maple and other software can run this in arbitrary precision from which one may recover analytical expressions for the solutions.


\begin{table}[h]
\centering
\setlength{\tabcolsep}{10pt}
\begin{tabular}{crl}
\hline
\hline
$|V|$ & $n$ & $\min\sqrt{\frac{\alpha}{|V|/\omega}}$ \\
        \hline
        $18$ & $380$ 		& $0.816$\\
        $19$ & $60$       	& $0.795$\\
        $20$ & $1214$		& $0.775$\\
        $21$ & $240$ 		& $0.756$\\
        $22$ & $816$ 		& $0.739$\\
        $23$ & $188$  		& $0.780$\\
        $24$ & $15506$ 		& $0.707$\\
        $25$ & $464$ 		& $0.775$\\
        $26$ & $4236$  		& $0.734$\\
        $27$ & $1434$  		& $0.745$\\
        $28$ & $25850$		& $0.732$\\
        $29$ & $1182$ 		& $0.719$\\
        $30$ & $46308$  	& $0.707$\\
        $31$ & $2192$ 		& $0.696$\\
        $32$ & $677402$ 	& $0.667$\\
        $33$ & $6768$  		& $0.625$\\
        $34$ & $132580$  	& $0.64$\\
        $35$ & $11150$ 		& $0.627$\\
        $36$ & $1963202$ 	& $0.615$\\
        $37$ & $14602$ 		& $0.604$\\
        $38$ & $814216$  	& $0.593$\\
        $39$ & $48462$  	& $0.632$\\
        $40$ & $13104170$ 	& $0.571$\\
        $41$ & $52488$ 		& $0.561$\\
        $42$ & $946226$	 	& $0.6$\\
        $43$ & $99880$  	& $0.635$\\
        $44$ & $39134640$ 	& $0.581$\\
        $45$ & $399420$  	& $0.571$\\
        $46$ & $34333800$ 	& $0.562$\\
        $47$ & $364724$ 	& $0.597$\\
\hline
\hline
\end{tabular}
\caption{The minimum value of $\sqrt{\frac{\alpha}{|V|/\omega}}$, which is a lower bound for $\eta_{\rm crit}$, for all vertex-transitive graphs with $|V| \le 47$ vertices. $n$ is the number of vertex-transitive graphs with the corresponding number of vertices.}
\label{tab:vt_min}
\end{table}


\subsection{Second strategy: Nonvertex-transitive graphs}
\label{secondstrat}


So far, we have focused on graphs that are vertex transitive. The reasons for this are that vertex-transitive graphs are relatively easy to identify and have a lot of symmetry. The latter is crucial for the optimization discussed in the second step of the method.

However, in \cite{Cabello21}, it is shown that there are other graphs leading to quantum correlations based on maximally entangled states violating a Bell inequality: Those admitting an orthonormal representation in dimension $\xi$ and nonnegative vertex weights $w=\{w_i\}_{i=1}^n$ such that $\sum_{i=1}^n w_i/\xi > \alpha(G,w)$, where $\alpha(G,w)$ is the independence number of the corresponding weighted graph.
In addition, in \cite{CKB15}, it is shown that a condition for these graphs is that the fractional chromatic number satisfies $\chi_f > \xi$. Interestingly, this condition, which is not sufficient for the graphs to have associated SI-C sets (see Theorem~1 in \cite{CKB15}) is, in fact, sufficient for having quantum correlations based on maximally entangled states violating a Bell inequality.

Therefore, another strategy to find examples with low $\eta_{\rm crit}$ would be the following: Find graphs with $\chi_f > \xi$. For each of them, find $w$, such that $\sum_{i=1}^n w_i/\xi > \alpha(G,w)$. Then, in a similar way as for Theorem~\ref{th:3}, we can prove for these graphs and weights,
\begin{equation}
\eta_{\rm crit} \le \sqrt{\frac{ \alpha(G,w)}{\sum_{i=1}^n \frac{w_i}{\xi}}}.
\end{equation}

Interestingly, since these weights are often natural numbers (e.g., for the Yu-Oh set \cite{YO12}, the weights are $2$ for $4$ of the vectors and $3$ for the other $9$ vectors \cite{Cabello21}), one can see the weighted graphs $(G,w)$ as nonweighted graphs $G$ with an extended
number of vertices (e.g., for the Yu-Oh set, $G$ would have $4 \times 2 + 9 \times 3 = 35$ vertices). Then, for finding candidates that may have a low $\eta_{\rm crit}$, we can use databases of nonweighted graphs of $13$ or more vertices (as it is known that the graphs for which $\chi_f > \xi$ must have, at least, $13$ vertices \cite{CKP16}) and there identify, first, graphs with $\chi_f > \omega$, where $\omega$ is the clique number, which is easier to compute than $\xi$. Since $\omega \le \xi$, this is a necessary condition. Later on, one can use existing computational tools~\cite{xu2021state} to obtain $\xi$.


\section{Noise and how to obtain better Bell inequalities}
\label{sec:no}


Here, we first show that, in all the examples with low $\eta_{\rm crit}$ presented so far, the values shown for $\eta_{\rm crit}$ are very sensitive to noise (Theorem~\ref{th:4}). 
The good news is that all the upper bounds for $\eta_{\rm crit}$ have been obtained with respect to a graph-based Bell inequality of the form in Eq.~\eqref{bbi}. However, for fixed correlations, we can optimize our Bell inequalities and obtain a {\em lower} value for $\eta_{\rm crit}$, which can also be more robust to noise. This optimization is discussed in Sec.~\ref{sec:opt} and uses the high symmetry of the measurements (their graph of orthogonality is vertex-transitive) producing the correlations.


\subsection{Noise and graph-based Bell inequalities}
\label{sec:noise}


In most discussions on the critical detection efficiency (e.g., \cite{Eberhard93,LS01,CRV08}), the effect of noise is modeled with the assumption that the effective state is of the form
\begin{equation}\label{eq:werner}
\rho=W |\psi\rangle \langle \psi | + (1-W) \frac{\openone}{d^2},
\end{equation}
where $|\psi\rangle$ is the targeted state, $W$ is the visibility of the state, $\openone$ is the identity, and $d$ is the dimension of the local system. In this work, we will follow this practice. However, it must be pointed out that, in some cases \cite{CFL05}, the effective state is not of the form \eqref{eq:werner} and it is important to take this into account \cite{BCCL06}.


\begin{theorem}
\label{th:3b}
For a Bell inequality of the form \eqref{bbi} associated to a graph $G(V,E)$ with vertex set $V$ and edge set $E$, and states of the form \eqref{eq:werner}, the critical visibility $W_{\rm crit}$, i.e., the minimal value of $W$ in $\eqref{eq:werner}$ for a violation in \eqref{bbi} is given by
\begin{equation}
\label{etaW}
W_{\rm crit} \le \frac{\alpha - Q_{\rm mix}}{(|V|/d) - Q_{\rm mix}},
\end{equation}
where 
\begin{equation}
\label{eq:thetamix}
Q_{\rm mix} = \frac{1}{d^2} \left(|V| - |E|/\Xi\right).
\end{equation}
\end{theorem}


\begin{proof}
Notice that $Q_{\rm mix}$ is the violation of \eqref{bbi} for the maximally mixed state.
\end{proof}


\begin{theorem}
\label{th:4}
For a Bell inequality of the form \eqref{bbi} associated to a graph $G(V,E)$ with vertex set $V$ and edge set $E$ and states of the form \eqref{eq:werner}, the critical detection efficiency $\eta_{\rm crit}$ is given by
\begin{equation}
\label{etaW}
        \eta_{\rm crit}^2 \le \frac{\alpha d^2}{|V| \left[W(d - 1) + 1\right] - |E| (1 - W)/\Xi}.
\end{equation}
\end{theorem}

\begin{proof}
The quantum violation of the Bell inequality \eqref{bbi} with state \eqref{eq:werner} and perfect detection efficiency is given by
\begin{equation}
Q' = W Q + (1-W) Q_{\text{mix}},
\end{equation}
where $Q \ge \frac{|V|}{d}$ is the expected quantum violation and $Q_{\rm mix}$ is the value for the maximally mixed state, which is given by \eqref{eq:thetamix}.
Then, the critical detection efficiency in Eq.~\eqref{eq:etac1} becomes Eq.~\eqref{etaW}.
\end{proof}


Theorem \ref{th:4} implies that, for the graph-based Bell inequalities \eqref{bbi}, $\eta_{\rm crit}$ rapidly increases with the number of edges unless $W$ is very close to $1$. Therefore, although experimental values of $W$ can be as high as $0.980$ for $d=3$ and $0.943$ for $d=17$ \cite{H20}, it would be desirable to find Bell inequalities for which the same correlations (i.e., the same state and the same measurements) have a value for $\eta_{\rm crit}$ (and for $W_{\rm crit}$) that is much less sensitive to noise. This problem is addressed in the next section.


\subsection{Optimized Bell inequalities based on symmetries}
\label{sec:opt}


\subsubsection{Introduction}


The graph-based Bell inequality \eqref{bbi} only takes into account probabilities $P(\Pi_i^A=a,\Pi_j^B=b)$ in which either $i=j$ or $i$ and $j$ are adjacent in the graph $G$.
However, in a Bell test Alice and Bob independently choose their measurements in such a way that the choice on one of them is spacelike separated from the recording of the measurement outcome on the other.
Therefore, while every observable $\Pi_i^A$ of Alice is compatible with every observable $\Pi_j^B$ of Bob, inequality \eqref{bbi} does not use most of the joint probability distributions $P(\Pi_i^A=a,\Pi_j^B=b)$. All the not used distributions are thus wasted.

An interesting question is the following: What if we use the same state and measurements used in the violation of the graph-based Bell inequality \eqref{bbi} and consider all the joint probability distributions $P(\Pi_i^A=a,\Pi_j^B=b)$? Can we then obtain a better Bell inequality? 

``Better'' may mean having higher resistance to noise (i.e., lower $W_{\rm crit}$), having lower critical detection efficiency $\eta_{\rm crit}$, or both, depending on what we are interested in. For designing a loophole-free Bell test, what we need is that the experimental values for the visibility and the detection efficiency, $W_{\text{exp}}$ and $\eta_{\text{exp}}$, respectively, are both above their respective critical values. That is, we need $W_{\text{exp}} > W_{\rm crit}$ and $\eta_{\text{exp}} > \eta_{\rm crit}$.

In the next sections we show how to compute $W_{\rm crit}$ and $\eta_{\rm crit}$ for states of interest. Intentionally, in the first example, we focus on an instance which does not offer neither a low $W_{\rm crit}$ nor a low $\eta_{\rm crit}$, but which will guide us to attack more interesting examples.


\begin{figure*}[t!] 
\centering
\includegraphics[width=.875\textwidth]{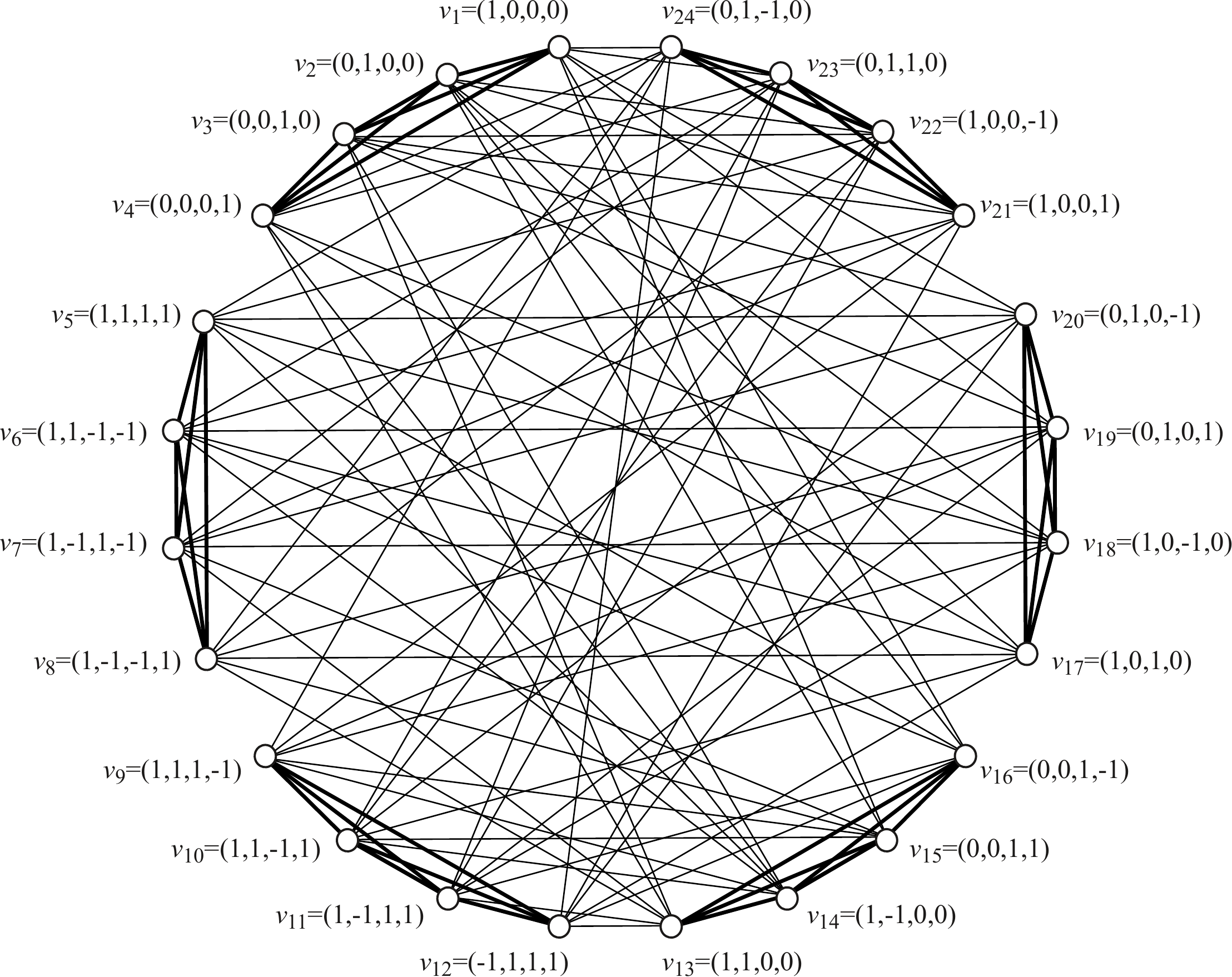}
\caption{Graph of orthogonality of the $24$ Pauli states in ${\cal P}_2(\mathbb{R})$. Vertices (dots) represent states and edges connect those that are orthogonal. The $24$ states can be distributed in $6$ disjoint orthogonal bases, which are indicated by thicker edges. \label{fig1}}
\end{figure*}


\subsubsection{Pauli-$24$}

Fig.~\ref{fig1} shows the graph of orthogonality of the $24$ (not normalized) states in ${\cal P}_2(\mathbb{R})$. This graph has $\alpha = 5$ and $\vartheta=\alpha^{\ast}=6$. Therefore, by preparing the maximally entangled state \eqref{eq:state} of local dimension $\xi=2^2=4$ and allowing each of the parties to choose between the $24$ two-outcome measurements represented by $|v_i\rangle \langle v_i|$, with $|v_i\rangle \in {\cal P}_2 (\mathbb{R})$, Alice and Bob can produce a violation of the Bell inequality \eqref{bbi} which, using Theorem~\ref{th:3}, allows us to conclude that
\begin{equation}
\label{etag}
\eta_{\rm crit}^{{\cal P}_2 (\mathbb{R})} \le 0.913
\end{equation}
and, using Theorem \ref{th:3b}, allows us to conclude that
\begin{equation}
\label{wg}
W_{\rm crit}^{{\cal P}_2 (\mathbb{R})} \le 0.911.
\end{equation}


\subsubsection{Gilbert's algorithm}


While our graph-based Bell inequalities are neither tight (i.e., facets of the local polytope) nor robust to noise, they can be further improved to offer better detection efficiency and noise robustness. This can be achieved by calculating the Bell functional in Eq.~\eqref{eq:bell_functional} using two different methods. 

The first method is a linear program which optimizes over the entire 
local polytope to find the best Bell functional~\cite{MBV21}. However, this technique requires enumerating and storing all the local deterministic points which are given by the vertices of the local polytope. This becomes an increasingly difficult computational task as the number of measurement settings increases. As an example, the local polytope corresponding to the Bell inequality with $24$ settings per party derived from $\mathcal{P}_2(\mathbb{R})$ has $2^{48}$ vertices. This number is too large to be stored on a standard computer. 
In this paper, we use a second method which is based on Gilbert's distance algorithm~\cite{G66} and does not require storing all the vertices of the local polytope. {However, it should be noted that, for the cases we are interested in (e.g., $\mathcal{P}_{2}(\mathbb{R})$), the problem is still intractable using Gilbert's original algorithm. The problem only becomes feasible when symmetries are also taken into account, as explained in Sec.~\ref{sss:gilbert_symm}.}

Gilbert's algorithm is a well-known numerical method to detect collisions between convex sets. It has been used 
for improving detection efficiencies of Bell inequalities~\cite{MBV21}, deciding whether or not a given correlation is nonlocal~\cite{MW19}, and entanglement witnessing~\cite{SG18,PSW20}. The algorithm minimizes the distance between a local point on a facet of the local polytope and a nonlocal point specified by the user. The minimization is achieved by iteratively finding a better local point that minimizes this distance. 
The algorithm terminates when the difference of distances between successive iterations falls below a certain threshold value (typically taken to be extremely small). The resulting Bell functional is then identified as the separating hyperplane between the specified nonlocal point and the local point on the facet found by minimizing the distance. 
{In the following, we will present this algorithm in more detail. The algorithm is based on an oracle which is capable of maximizing over the local polytope $\mathcal{P}$ the inner product between a point in $\mathbb{R}^{n}$. Initially, one has to specify the local polytope $\mathcal{P} \subset \mathbb{R}^{n}$, presented as the convex hull of vertices $\lbrace c_{k} \rbrace_{k}$, and a point $q \in \mathbb{R}^{n}$ associated to the given quantum correlation. Then, the algorithm proceeds as follows. First, it chooses a point $s_{0} \in \mathcal{P}$. Second, with the given point $s_{k}$, it uses the oracle to compute
\begin{align} \label{eq:gilbert_iteration}
    \Tilde{s}_{k} :&= \underset{p \in \mathcal{P}} { \text{argmax}}\, \langle q-s_{k}, p-s_{k} \rangle\nonumber\\
    &= \underset{p \in \mathcal{P}} { \text{argmax}}\, \langle q-s_{k}, p \rangle.
\end{align}
Third, given $s_{k}$ and $\Tilde{s}_{k}$, it calculates the convex combination of both which minimizes the distance to the quantum point $q$, that is,
\begin{align}
    \lambda_k := \underset{\lambda \in [0,1]}{\text{min}} \, \vert \vert (1-\lambda) s_{k} + \lambda \Tilde{s}_{k} - q \vert \vert.
\end{align}
Since the objective function in \eqref{eq:gilbert_iteration} is linear and $\mathcal{P}$ convex, the maximizer will be an extreme point of $\mathcal{P}$. In particular, $\Tilde{s}_{k}$ will be a vertex of the local polytope $\mathcal{P}$. 
The optimal value for $\lambda$ in the $k$-th iteration, denoted by $\lambda_{k}$, can be computed directly and is given by 
\begin{align}
    \lambda_{k} = \text{min} \, \left\lbrace \frac{\langle q-s_{k}, \Tilde{s}_{k}- s_{k} \rangle}{\vert \vert \Tilde{s}_{k} -s_{k} \vert \vert^{2}} , 1 \right\rbrace.
\end{align}
Then, it defines the new starting point to be $s_{k+1} := (1-\lambda_{k})s_{k} + \lambda_{k} \Tilde{s}_{k}$ and it proceeds with the second step. 
{In the standard Gilbert's algorithm, the oracle is implemented by enumerating all the vertices of the local polytope $\mathcal{P}$ to compute the inner product in the last equality of Eq.~\eqref{eq:gilbert_iteration}.
}
For a geometrical interpretation of the iteration of Gilbert's algorithm, see Fig.~\ref{generic_gilbert}.}


\begin{figure}[h!] 
\centering
\includegraphics[width=.45\textwidth]{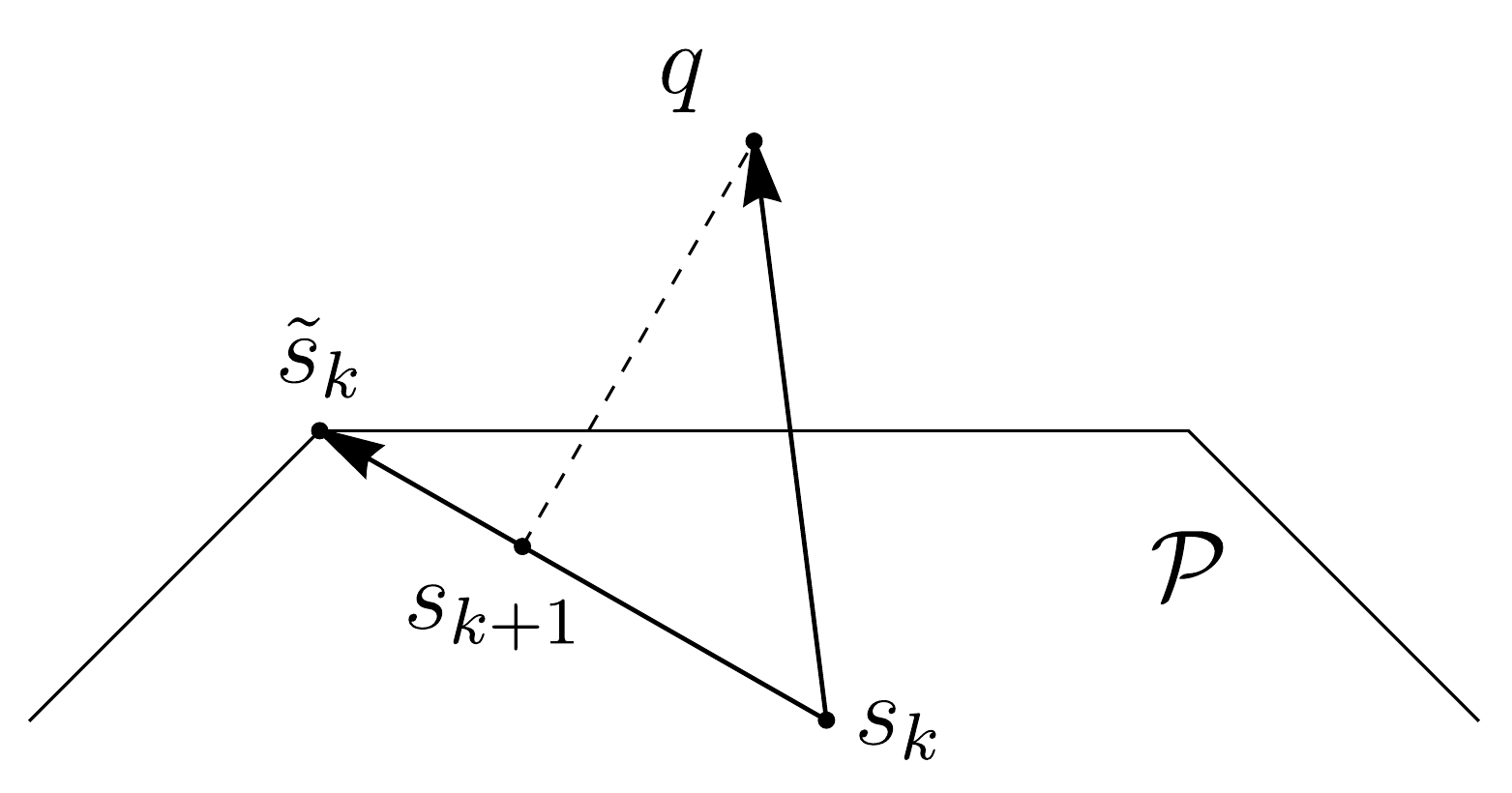}
\caption{Illustration of Gilbert's algorithm. The quantum point $q$ lies outside the local polytope $\mathcal{P}$. Starting with an arbitrary local point $s_{k} \in \mathcal{P}$, the oracle yields the point $\Tilde{s}_{k}$ within $\mathcal{P}$, maximizing the overlap with $q-s_{k}$. From there, a new starting point $s_{k+1}$ can be calculated. \label{generic_gilbert}}
\end{figure}


This algorithm provides a sequence of Bell functionals, which become better with each iteration. Note that it does not necessarily provide a tight Bell inequality like the first method.
Moreover, calculating the local bound of the resultant Bell functional still remains an NP hard problem, which again requires enumerating and storing all the local deterministic points 
at least for one party\footnote{For a given bipartite Bell functional where each measurement has only two outcomes $\pm 1$, once we have fixed the deterministic assignment of one party, the optimal value of the Bell functional and the corresponding deterministic assignments of another party are determined. Refer to Sec. 1 in appendix C of Ref.~\cite{gonzales2022optimal} for more details.}.
This issue is also shared by the oracle in the standard Gilbert's algorithm, since the evaluation of Eq.~\eqref{eq:gilbert_iteration} is equivalent to find the local bound of one Bell functional.


\subsubsection{Gilbert's algorithm with symmetrization} \label{sss:gilbert_symm}


{The two problems mentioned at the end of the previous section can be solved simultaneously by replacing the hitherto generic oracle with an oracle that takes into account the symmetries of the problem. This provides an advantage that is crucial for improving the speed of convergence.} In this respect, it is important to notice that the vertices of the local polytope (and the polytope itself) are invariant under the following invertible transformations: (i)~Swapping the outcomes of a measurement setting for either Alice or Bob. (ii)~Simultaneously permuting the measurement settings of Alice and Bob. (iii)~Swapping the measurement settings of Alice and Bob. {Here, ``invariant'' means that the transformations map local correlations to local correlations.}

The joint probability distributions obtained by performing measurements on a quantum state are also invariant under some of these transformations. We denote by $\mathcal{S}$ a subset of these transformations which keep the quantum joint probability distributions and the local
polytope invariant simultaneously.

We consider a matrix similar to the one in Eq.~\eqref{eq:bell_functional} in which the entries are probability distributions instead of coefficients, given as,
\begin{align}
&M_{p} =\nonumber\\ &\left(
\setlength{\tabcolsep}{0.1em}
\begin{tabular}{c|ccc}
$1$ & $P(\Pi_1^A\!=\!1)$ & \dots & $P(\Pi_m^A\!=\!1)$ \\
\hline 
$P(\Pi_1^B\!=\!1)$ & $P(\Pi_1^A\!=\!\Pi_1^B\!=\!1)$ & $\dots$ & $P(\Pi_m^A\!=\!\Pi_1^B\!=\!1)$ \\
$\vdots$ & $\vdots$ & $\ddots$ & $\vdots$ \\
$P(\Pi_m^B\!=\!1)$ & $P(\Pi_1^A\!=\!\Pi_m^B\!=\!1)$ & $\dots$ & $P(\Pi_m^A\!=\!\Pi_m^B\!=\!1)$
\end{tabular}
\right).
\label{eq:bell_prob}
\end{align}
The corresponding Bell inequality can then be calculated as $\tr(I M_p^T) \leq \lambda$. 
Under the transformations $S\in \mathcal{S}$ we have,
\begin{equation}\label{eq:symm1}
\tr(I M_p^T)\! =\! \tr(I S(M_p)^T)\! =\! \tr(S^{-1}(I) M_p^T) \leq \lambda,
\end{equation}
where $S(M_p)$ is the resultant matrix after transformation $S$.
Therefore, we have
\begin{align}
\tr(I M_p^T)\! =\! \tr\left(I \bar{M}_p^T\right)\! =\! \tr(\bar{I} M_p^T)\! =\! \tr(\bar{I} \bar{M}_p^T),
\label{eq:symmetrization}
\end{align}
where 
\begin{equation}
\bar{M}_p = \frac{1}{|\mathcal{S}|}\sum_{S_i\in \mathcal{S}} S_i(M_p),\ \bar{I} = \frac{1}{|\mathcal{S}|}\sum_{S_i\in \mathcal{S}} S_i^{-1}(M_p),
\end{equation}
and $|\mathcal{S}|$ is the cardinality of $\mathcal{S}$.
{Using this symmetry $\mathcal{S}$, it is sufficient to consider only inequalities that share this symmetry. In particular, for a symmetric inequality, the vertices of the local polytope can be partitioned into different equivalence classes with respect to that symmetry and the local bound can be computed by choosing only \textit{one} representative out of each class. This drastically reduces the total number of local vertices required, allowing us to enumerate all symmetrized local points and evaluate the local bound. }

{For convenience, here we focus on the symmetrization applied only to  Alice's measurement settings. Each equivalence class consists of vertices which can be transformed into each other by using the aforementioned symmetry transformations, while the same is not true for vertices in different equivalence classes, i.e., the partition generates disjoint sets. This leads to a modified oracle, which is much more efficient than the original one since the number of equivalence classes could much smaller than the number of all vertices. This is indeed the case for $\mathcal{P}_{2}(\mathbb{R})$. 

Now, we describe how to obtain the reduced set of vertices on which the optimization in \eqref{eq:gilbert_iteration} has to run. In the first step, we have to determine the allowed symmetry transformations $\mathcal{S}$, which should be shared by the chosen quantum point $q$ and the local polytope $\mathcal{P}$. As already pointed out, $\mathcal{P}$ is invariant under the permutation of parties, the permutation of the measurements for each party, and the permutation of the outcomes for each measurement. By construction, the chosen quantum point $q$ is also invariant under the permutation of parties. Usually, $q$ can change after the permutation of outcomes for each measurement. In general, determining the permutation symmetries of measurements in the point $q$ can be difficult when the number of measurements is large. In the cases considered here, those symmetry transformations correspond to the ones in the automorphism group of the graph associated to the SI-C sets. Therefore, the symmetry transformations used in the Gilbert's algorithm with symmetrization are the ones in the automorphism group and the permutation of parties. It remains to explain how the vertices are partitioned into different equivalence classes. 

The first important observation is that now we do not need to generate, store, and classify all vertices, since assignments with different number of $1$'s cannot be equivalent to each other. Hence, we can do the classification inductively. We start with the assignment that only contains $0$'s. Obviously, this is invariant under all possible permutations. From this, we generate all possible assignments that can be obtained by replacing one $0$ by one $1$. Within this set, we check whether some of these assignments are equivalent under the given symmetry transformations $\mathcal{S}$, which are presented as permutations. Then, we only keep one representative for each class. This procedure is repeated until no $0$ is left in the assignment vector.}

Selecting only a single vertex from each equivalence class (and all the vertices of Bob), we find that the total number of deterministic
assignments for Alice is $21564$ for the Bell inequality corresponding to $\mathcal{P}_2(\mathbb{R})$
(while, without symmetrization, it would be $2^{24}$).

We also modify the Gilbert's algorithm to evaluate the Bell functional according to the symmetrization 
procedure in Eq.~\eqref{eq:symmetrization}. Specifically, we symmetrize the local point chosen in each iteration of the program after minimizing its distance from the target nonlocal point, see Fig.~\ref{generic_gilbert_s} for a simple illustration. This results in better convergence times of the algorithm since symmetrization does not increase the distance. 



\begin{figure}[h!] 
\centering
\includegraphics[width=.45\textwidth]{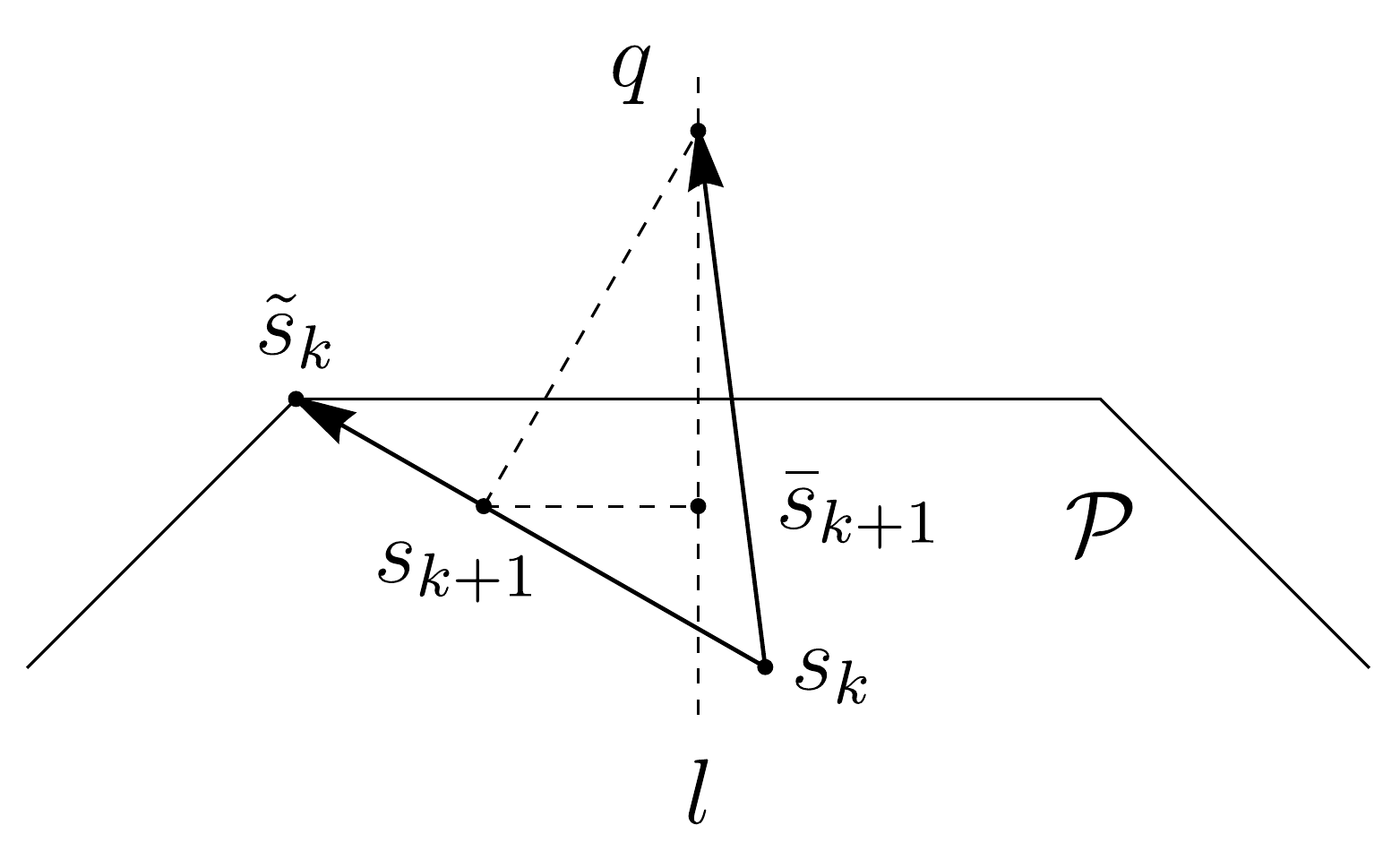}
\caption{Illustration of Gilbert's algorithm with symmetrization. In this simple example, the point $q$ and the polytope $\mathcal{P}$ are invariant under the flip around line~$l$. After the point $s_{k+1}$ has been found in the standard Gilbert's algorithm, we obtain its symmetrization $\bar{s}_{k+1}$ for the flip around line~$l$. The point $\bar{s}_{k+1}$ is used instead of $s_{k+1}$ as the new starting point for the next iteration. \label{generic_gilbert_s}}
\end{figure}

\subsubsection{Example: Bell inequality for Pauli-$24$ that minimizes $\eta_{\rm crit}$ and $W_{\rm crit}$}


As an example, here we present a Bell functional for the correlations produced by the maximal entangled state \eqref{eq:state}, with $\xi=4$, and the measurements associated to the states in $\mathcal{P}_2(\mathbb{R})$ which has lower detection efficiency and larger tolerance to white noise than those of to the graph-based Bell functional \eqref{bbi}. The purpose of this example is to show that the values of $\eta_{\rm crit}$ and $W_{\rm crit}$ obtained in the first step of the method can, in principle, always be improved. 


Applying Gilbert's algorithm and the symmetrization technique to the correlations produced by ${\cal P}_2(\mathbb{R})$ and maximally entangled states, we have obtained the following Bell inequality:
\begin{equation}
\label{eq:opP24}
I_{{\cal P}_2(\mathbb{R})} \le 0,
\end{equation}
where

\begin{equation}
\label{ec:p24-marg}
I_{{\cal P}_2(\mathbb{R})} = \left(
\begin{array}{c|rr}
& |v\rangle & |v\rangle\\ \hline
\langle v| & M_1 & M_2\\
\langle v| & M_2 & M_1
\end{array}
\right),
\end{equation}
with $M_1$ to be a matrix full of $0$'s, $|v\rangle$ to be a vector full of $-6$'s, and
\begin{equation}
\label{ec:p24-marg}
M_2 = \left(
\setlength{\tabcolsep}{0.1em}
\begin{array}{rrrrrrrrrrrrrrrrrrrrrrrr}
5 & 5 & \bar{4} & \bar{4} & 5 & 5 & \bar{4} & \bar{4} & 5 & 5 & \bar{4} & \bar{4}\\
5 & 5 & \bar{4} & \bar{4} & \bar{4} & \bar{4} & 5 & 5 & \bar{4} & \bar{4} & 5 & 5\\
\bar{4} & \bar{4} & 5 & 5 & 5 & 5 & \bar{4} & \bar{4} & \bar{4} & \bar{4} & 5 & 5\\
\bar{4} & \bar{4} & 5 & 5 & \bar{4} & \bar{4} & 5 & 5 & 5 & 5 & \bar{4} & \bar{4}\\
5 & \bar{4} & 5 & \bar{4} & 5 & \bar{4} & 5 & \bar{4} & 5 & \bar{4} & 5 & \bar{4}\\
5 & \bar{4} & 5 & \bar{4} & \bar{4} & 5 & \bar{4} & 5 & \bar{4} & 5 & \bar{4} & 5\\
\bar{4} & 5 & \bar{4} & 5 & 5 & \bar{4} & 5 & \bar{4} & \bar{4} & 5 & \bar{4} & 5\\
\bar{4} & 5 & \bar{4} & 5 & \bar{4} & 5 & \bar{4} & 5 & 5 & \bar{4} & 5 & \bar{4}\\
5 & \bar{4} & \bar{4} & 5 & 5 & \bar{4} & \bar{4} & 5 & \bar{4} & 5 & 5 & \bar{4}\\
5 & \bar{4} & \bar{4} & 5 & \bar{4} & 5 & 5 & \bar{4} & 5 & \bar{4} & \bar{4} & 5\\
\bar{4} & 5 & 5 & \bar{4} & 5 & \bar{4} & \bar{4} & 5 & 5 & \bar{4} & \bar{4} & 5\\
\bar{4} & 5 & 5 & \bar{4} & \bar{4} & 5 & 5 & \bar{4} & \bar{4} & 5 & 5 & \bar{4}
\end{array}
\right),
\end{equation}
where $\bar{4} = -4$.

For the maximally entangled state \eqref{eq:state}, the quantum value is
\begin{equation}
I_{{\cal P}_2(\mathbb{R})}= 18.
\end{equation}
In addition, for the correlations produced by ${\cal P}_2(\mathbb{R})$ and maximally entangled states,
\begin{equation}
W_{\rm crit} = \frac{7}{9} = 0.778,
\end{equation}
which is $14.62\%$ lower than the upper bound in \eqref{wg}, and
\begin{equation}
\eta_{\rm crit}= \frac{4}{5} = 0.8,
\end{equation}
which is $12.38\%$ lower than the upper bound in \eqref{etag}.

We can prove that $W_{\rm crit}$ and $\eta_{\rm crit}$ are the smallest possible values for the correlations produced by ${\cal P}_2(\mathbb{R})$ and maximally entangled states as follows: First, we collect all the $452929$ pairs of deterministic assignments for the two parties which achieve the maximal bound for LHV models. Each of them corresponds to a matrix $M_p$. After symmetrization, there are only $132$ different matrices $\bar{M}_p$, whose convex combination can lead to the corresponding quantum probability matrix either with $W=7/9$ or with $\eta = 4/5$, as one can verify with linear programming. 
Inequality \eqref{eq:opP24} is not tight. There is a tight Bell inequality providing the same $\eta_{\rm crit}$ and $W_{\rm crit}$ than the ones for inequality \eqref{eq:opP24}, but it does not have the two blocks of $0$'s that we have in \eqref{ec:p24-marg}. If we want to keep the $0$'s, inequality \eqref{eq:opP24} is the only solution. 

Notice that symmetries of the initial graph are crucial for finding \eqref{ec:p24-marg}. For example, notice that there are only $6$ different parameters in the symmetric inequality for the case of $\mathcal{P}_2(\mathbb{R})$. In comparison, there are $624$ parameters in the nonsymmetric inequality for $\mathcal{P}_2(\mathbb{R})$. In the general case, there are $2m+2$ parameters in the symmetric inequality for $\mathcal{P}_m(\mathbb{R})$ and $\mathcal{P}_m(\mathbb{C})$. This makes it also possible to find a better inequality without resorting to Gilbert's method. To be more explicit, we can choose $t$ different values for each parameter, then there are $t^{2m+2}$ different inequalities. For each inequality, we can verify whether it separates the target quantum point and the LHV polytope or not by considering only the deterministic assignments up to symmetry. As discussed before, when $m=2$, there are $21564$ different deterministic assignments for Alice up to symmetry. Therefore, for a fixed inequality, this verification can be done very fast. Similarly, we can calculate $\eta_{\rm crit}$ and $W_{\rm crit}$ for each inequality. As we can see in Eq.~\eqref{ec:p24-marg}, we can set some parameters to be $0$ for the best $\eta_{\rm crit}$ and $W_{\rm crit}$. This trick can speed up the numerical calculation further.


\section{Towards high-dimensional long-distance loophole-free Bell tests}
\label{sec:alg}


We have shown how to identify quantum correlations between systems of moderate dimension ($d \le 128$) that allow for loophole-free Bell nonlocality with low detection efficiency. Our results imply that, probably, loophole-free Bell nonlocality can be achieved in carefully designed tests involving pairs of systems of these dimensions, which is interesting by itself as it goes beyond previous loophole-free Bell tests, all of them based on qubits. More interestingly, our results also imply that loophole-free Bell nonlocality can be achieved through longer distances than those of previous loophole-free Bell tests. When photons propagate thorough fibers, they experience propagation losses proportional to the propagation distance and which depend on the optical wavelength. This means that in a Bell test over long distances (and unless we add, e.g., heralded qudit amplifiers \cite{GPS10} or split each photon into two \cite{CS12,MMGSFHCRJ16} before the local measurements), the detection efficiency decreases with the distance. For example, with telecom wavelengths, in $10$ km of fiber, we may have losses of $0.2$ dB/km, which implies multiplying by $0.64$ the detection efficiency that we had before adding $10$ km of fiber. Therefore, if we have examples of loophole-free Bell nonlocality requiring $\eta_{\rm crit} < 0.5$, we can achieve loophole-free nonlocality over $10$ km if we have $\eta_{\rm exp} >0.785$ before adding the $10$ km of fiber. Such $\eta_{\rm exp}$ has been achieved in previous photonic loophole-free experiments \cite{GVW15,SMC15}, even including ($<200$ m of) fibers and the couplings. 

In this paper, we have shown how to achieve $\eta_{\rm crit} < 0.52$ with local dimensions $16$ [see Eq.~\eqref{eq:4320}] and how to obtain Bell inequalities with even lower $\eta_{\rm crit}$ and higher resistance to noise. Now the question is what are the values of the visibility $W$ that are experimentally achievable for the required configurations. Although some previous results are very promising \cite{H20}, it is not clear to us whether similar values ($W_{\rm exp} > 0.95$) can be achieved for the states and measurements described in this paper. For further progress, we need to know what pairs $(\eta_{\rm exp}, W_{\rm exp})$ can be obtained for the type of states and measurements proposed here (e.g., for $d=16$ and Pauli-$4320$). In addition, we have introduced methods to identify further examples with low $\eta_{\rm crit}$ requiring smaller number of settings. 


\begin{figure}[t!] 
\centering
\includegraphics[width=.45\textwidth]{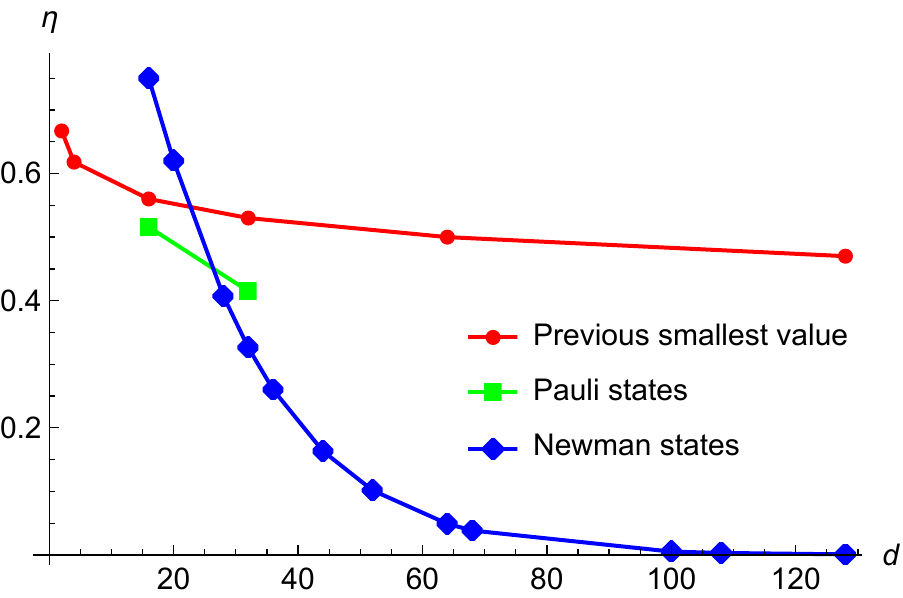}
\caption{$\eta_{\rm crit}$ as a function of the dimension $d$ of the local system. The previous smallest values are those in \cite{Eberhard93,VPB10,M22} and use nonmaximally entangled states. Pauli (Newman) states refer to the case in which the local measurements are projectors on the Pauli (Newman) states and the initial state is maximally entangled, as described in Sec.~\ref{pnstates}. \label{fig2}}
\end{figure}


As shown in Fig.~\ref{fig2}, the first step of our method already yields upper bounds for $\eta_{\rm crit}$ which are substantially smaller than the lowest values previously known for any dimension $d\ge 16$. These values indicate that there are quantum correlations which have sufficiently low $\eta_{\rm crit}$ for loophole-free Bell tests with higher-dimensional quantum systems and, eventually, over longer distances. We have also shown how to improve any of these values and we have described how we are trying to find more, and probably better, examples in the future. 

It now remains to be seen how far we can go on the experimental side. In particular, there is one question, to which we do not have the answer: Is there a way to encode high-dimensional entanglement in photons that allows, at the same time, (i) to measure all the necessary one-dimensional projectors locally and switch quickly between them, (ii) distribute the photons over, e.g., $5$--$10$ km achieving visibilities $W_{\rm exp}>0.95$, (iii) using superconducting detectors to actually achieve $\eta_{\rm exp}>0.5$? Hopefully, the combination of techniques presented here and the interaction with experimental groups will help us to produce high-dimensional loophole-free Bell nonlocality over long distances in the near future.


The code supporting the results reported in this paper is available in~\cite{code_rep}.


\section*{Acknowledgments}

{
We thank Ernesto Galv\~ao, Junior R.\ Gonzales-Ureta, Petr Lison\v{e}k, Debashis Saha, and Stefan Trandafir for discussions, and the University of Siegen for enabling our computations through the OMNI cluster. This work is supported by \href{http://dx.doi.org/10.13039/100009042}{Universidad de Sevilla} Project Qdisc (Project No.\ US-15097), with FEDER funds, \href{http://dx.doi.org/10.13039/501100011033:}{MCINN/AEI} Projet No.\ PID2020-113738GB-I00, QuantERA grant SECRET, by \href{http://dx.doi.org/10.13039/501100011033:}{MCINN/AEI} (Project No.\ PCI2019-111885-2),
Deutsche Forschungsgemeinschaft (DFG, German Research Foundation, Project No.\ 447948357 and No.\ 440958198), the Sino-German Center for Research Promotion (Project No.\ M-0294), and the ERC (Consolidator Grant 683107/TempoQ). ZPX is supported by the Humboldt Foundation. JS is supported by the House of Young Talents of the University of Siegen. JRP is supported by \href{http://dx.doi.org/10.13039/100009042}{Universidad de Sevilla} Project No.\ US-1254251, with FEDER funds, and Junta de Andaluc\'{\i}a Project No.\ P20-00592.
}


\end{document}